\def \stoc{0}

\ifnum\stoc=1
\def \deffontsize{11pt}
\def \defhmargin{1in}
\def \defvmargin{1in}
\else
\def \deffontsize{11pt}
\def \defhmargin{1in}
\def \defvmargin{2cm}
\fi

\documentclass[xelatex, \deffontsize]{article}

\def\todotoggle{0}

\providecommand{\deffontsize}{11pt} 
\providecommand{\todotoggle}{1} 

\usepackage{amsmath, amssymb, amsthm, amsfonts, bbm}
\usepackage[tbtags]{mathtools}
\usepackage{graphicx}
\usepackage{verbatim}
\usepackage{enumerate}
\usepackage{palatino}
\usepackage[usenames,dvipsnames]{color}
\usepackage{units}
\usepackage{tikz}
\usepackage{subfigure}
\usepackage[utf8]{inputenc}
\usepackage[T1]{fontenc}
\usepackage{bm}
\usepackage[pdfstartview=FitH, pdfpagemode=None, colorlinks=true, citecolor=blue, linkcolor=blue, urlcolor=magenta, pagebackref=true, unicode=false]{hyperref}

\providecommand{\defvmargin}{1in} 
\providecommand{\defhmargin}{1in} 
\usepackage[top=\defvmargin, bottom=\defvmargin, left=\defhmargin, right=\defhmargin]{geometry}

\usepackage{xargs} 
\usepackage[colorinlistoftodos,prependcaption,textsize=footnotesize]{todonotes}

\newcommand{\mytodo}[1]{\ifnum\todotoggle=1{#1}\fi}

\newcommandx{\unsure}[2][1=]{\mytodo{\todo[linecolor=blue,backgroundcolor=blue!25,bordercolor=blue,#1]{#2}}}
\newcommandx{\change}[2][1=]{\mytodo{\todo[linecolor=cyan,backgroundcolor=cyan!25,bordercolor=cyan,#1]{#2}}}
\newcommandx{\info}[2][1=]{\mytodo{\todo[linecolor=green,backgroundcolor=green!25,bordercolor=green,#1]{#2}}}
\newcommandx{\improve}[2][1=]{\mytodo{\todo[linecolor=red,backgroundcolor=red!25,bordercolor=red,#1]{#2}}}
\newcommandx{\thiswillnotshow}[2][1=]{\todo[disable,#1]{#2}}
\newcommand{\tableoftodos}{\ifnum\todotoggle=1 \listoftodos[Comments/To Do's] \fi}

\usetikzlibrary{positioning,snakes,mindmap,calc,fit,arrows,shapes}

\tikzstyle{box} = [shape=rectangle,draw=black,thick]

\newcommand{\myignore}[1]{}

\newcommand{\TealBlue}[1]{\textcolor{TealBlue}{#1}}
\newcommand{\Peach}[1]{\textcolor{Peach}{#1}}
\newcommand{\Cyan}[1]{\textcolor{cyan}{#1}}
\newcommand{\Red}[1]{\textcolor{red}{#1}}
\newcommand{\Navy}[1]{\textcolor{Blue}{#1}}
\newcommand{\Blue}[1]{\textcolor{Blue}{#1}}
\newcommand{\Green}[1]{\textcolor{OliveGreen}{#1}}
\newcommand{\Black}[1]{\textcolor{black}{#1}}
\newcommand{\White}[1]{\textcolor{white}{#1}}
\newcommand{\Code}[1]{\Cyan{\texttt{#1}}}
\newcommand{\Gray}[1]{\textcolor{gray}{#1}}
\newcommand{\Magenta}[1]{\textcolor{magenta}{#1}}

\newcommand{\tTealBlue}[1]{\ifnum\todotoggle=1\TealBlue{#1}\else{#1}\fi}
\newcommand{\tPeach}[1]{\ifnum\todotoggle=1\Peach{#1}\else{#1}\fi}
\newcommand{\tCyan}[1]{\ifnum\todotoggle=1\Cyan{#1}\else{#1}\fi}
\newcommand{\tRed}[1]{\ifnum\todotoggle=1\Red{#1}\else{#1}\fi}
\newcommand{\tNavy}[1]{\ifnum\todotoggle=1\Navy{#1}\else{#1}\fi}
\newcommand{\tBlue}[1]{\ifnum\todotoggle=1\Blue{#1}\else{#1}\fi}
\newcommand{\tGreen}[1]{\ifnum\todotoggle=1\Green{#1}\else{#1}\fi}
\newcommand{\tBlack}[1]{\ifnum\todotoggle=1\Black{#1}\else{#1}\fi}
\newcommand{\tWhite}[1]{\ifnum\todotoggle=1\White{#1}\else{#1}\fi}
\newcommand{\tCode}[1]{\ifnum\todotoggle=1\Code{#1}\else{#1}\fi}
\newcommand{\tGray}[1]{\ifnum\todotoggle=1\Gray{#1}\else{#1}\fi}
\newcommand{\tMagenta}[1]{\ifnum\todotoggle=1\Magenta{#1}\else{#1}\fi}

\newcommand{\inbrace}[1]{\left \{ #1 \right \}}
\newcommand{\inparen}[1]{\left ( #1 \right )}
\newcommand{\insquare}[1]{\left [ #1 \right ]}
\newcommand{\inangle}[1]{\left \langle #1 \right \rangle}
\newcommand{\infork}[1]{\left \{ \begin{matrix} #1 \end{matrix} \right .}
\newcommand{\inmat}[1]{\begin{matrix} #1 \end{matrix}}
\newcommand{\inbmat}[1]{\begin{bmatrix} #1 \end{bmatrix}}

\newcommand{\inabs}[1]{\begin{vmatrix} #1 \end{vmatrix}}
\newcommand{\norm}[2]{\begin{Vmatrix} #2 \end{Vmatrix}_{#1}}

\newcommand{\set}[1]{\inbrace{#1}}
\newcommand{\setdef}[2]{\set{#1 : #2}}
\newcommand{\defeq}{\stackrel{\mathrm{def}}{=}}

\DeclareMathOperator*{\expect}{\mathbb{E}}
\DeclareMathOperator*{\Ex}{\mathbb{E}}
\DeclareMathOperator*{\Var}{\mathrm{Var}}

\newcommand{\sAND}{\ \mathrm{and} \ }

\newcommand{\half}{\nicefrac{1}{2}}
\newcommand{\third}{\nicefrac{1}{3}}
\newcommand{\fourth}{\nicefrac{1}{4}}

\newcommand{\sbit}{\set{1,-1}}

\newcommand{\wtilde}[1]{\widetilde{#1}}
\newcommand{\what}[1]{\widehat{#1}}

\newcommand{\dTV}{d_{\mathrm{TV}}}

\newcommand{\bbE}{{\mathbb E}}

\newcommand{\bbN}{{\mathbb N}}

\newcommand{\bbR}{{\mathbb R}}

\newcommand{\bbZ}{{\mathbb Z}}

\newcommand{\bx}{\mathbf{x}}
\newcommand{\by}{\mathbf{y}}

\newcommand{\bsigma}{{\boldsymbol \sigma}}

\newcommand{\calA}{\mathcal{A}}
\newcommand{\calB}{\mathcal{B}}

\newcommand{\calG}{\mathcal{G}}

\newcommand{\calN}{\mathcal{N}}

\newcommand{\calU}{\mathcal{U}}
\newcommand{\calV}{\mathcal{V}}
\newcommand{\calX}{\mathcal{X}}

\newcommand{\eps}{\varepsilon}

\newcommand{\poly}{\mathrm{poly}}

\newcommand{\Inf}{\mathrm{Inf}}

\newcommand{\Supp}{\mathrm{Supp}}

\newcommand{\naive}{na\"{i}ve }

\newtheorem{theorem}{Theorem}[section]
\newtheorem{defn}[theorem]{Definition}
\newtheorem{definition}[theorem]{Definition}

\newtheorem{lem}[theorem]{Lemma}
\newtheorem{proposition}[theorem]{Proposition}
\newtheorem{cor}[theorem]{Corollary}
\newtheorem{fact}[theorem]{Fact}
\newtheorem{problem}[theorem]{Problem}

\newtheorem{claim}[theorem]{Claim}

\newtheorem{remark}[theorem]{Remark}

\newenvironment{proof-overview}{\noindent {\em Proof Overview.} \hspace*{1mm}}{\hfill $\Box$}
\newenvironment{proofof}[1]{\noindent {\em Proof of #1.} \hspace*{1mm}}{\hfill $\Box$}
\newenvironment{proof-sketch}{\noindent {\em Sketch of Proof.} \hspace*{1mm}}{\hfill $\Box$}
\newenvironment{proof-sketch-of}[1]{\noindent {\em Proof Sketch of #1.} \hspace*{1mm}}{\hfill $\Box$}
\newenvironment{proof-idea}{\noindent{\em Proof Idea.}\hspace*{1mm}}{\qed\bigskip}
\newenvironment{proof-attempt}{\noindent{\em Proof Attempt.}\hspace*{1mm}}{\qed\bigskip}

\newcommand{\Cov}{{\rm Cov}}
\newcommand{\Ber}{{\rm Ber}}

\newcommand{\DSBS}{\mathrm{DSBS}}
\newcommand{\ANIS}{\text{\sc Gap-Non-Int-Sim}}
\newcommand{\ABMIP}{\text{\sc Gap-Bal-Max-Inner-Product}}

\newcommand{\fullver}[2]{\ifnum\stoc=1{#1}\else{#2}\fi}
\newcommand{\noapx}[2]{\ifnum\apx=0{#1}\else{#2}\fi}
\def \stocapx{0} 

\begin{document}

\fullver{\setcounter{page}{0}}{\setcounter{page}{-1}}

\title{Decidability of Non-Interactive Simulation of Joint Distributions}

\author{
	Badih Ghazi\thanks{Computer Science and Artificial Intelligence Laboratory, Massachusetts Institute of Technology, Cambridge MA 02139. Supported in part by NSF CCF-1420956, NSF CCF-1420692 and CCF-1217423.  {\tt badih@mit.edu}.} 
	\and 
	Pritish Kamath\thanks{Computer Science and Artificial Intelligence Laboratory, Massachusetts Institute of Technology, Cambridge MA 02139. Supported in part by NSF CCF-1420956 and NSF CCF-1420692.  {\tt pritish@mit.edu}.} 
	\and 
	Madhu Sudan\thanks{Harvard John A. Paulson School of Engineering and Applied Sciences. Part of this work was done while at Microsoft Research New England. Supported in part by NSF Award CCF 1565641. {\tt madhu@cs.harvard.edu}.}%
}
\date{\today}
\maketitle

\thispagestyle{empty}

\abstract{
We present decidability results for a sub-class of ``non-interactive'' simulation
problems, a well-studied class of problems in information theory.
A {\em non-interactive simulation} problem is specified by two distributions $P(x,y)$ and $Q(u,v)$: The goal is to determine if two players, 
Alice and Bob, that observe sequences $X^n$ and $Y^n$ respectively where $\set{(X_i, Y_i)}_{i=1}^n$ are drawn i.i.d. from $P(x,y)$ can generate pairs $U$ and $V$ respectively  (without communicating with each other) with a joint distribution that is arbitrarily close in total variation to $Q(u,v)$.
Even when $P$ and $Q$ are extremely simple: e.g., $P$ is uniform on the triples
$\{(0,0), (0,1), (1,0)\}$ and $Q$ is a ``doubly symmetric binary source'', i.e., $U$ and $V$ are uniform $\pm 1$ variables with correlation say $0.49$, it is open if $P$ can simulate $Q$.

In this work, we show that whenever $P$ is a distribution on a finite domain and
$Q$ is a $2 \times 2$ distribution, then the non-interactive simulation
problem is {\em decidable}: specifically, given $\delta > 0$ the algorithm runs in time bounded by some function of $P$ and $\delta$ and either gives a non-interactive simulation
protocol that is $\delta$-close to $Q$ or asserts that no protocol gets
$O(\delta)$-close to $Q$. The main challenge to such a result is determining explicit (computable) convergence bounds on the number $n$ of samples that need to be drawn from $P(x,y)$ to get $\delta$-close to $Q$. We invoke contemporary results 
from the analysis of Boolean functions such as the invariance principle and a regularity lemma to obtain such explicit bounds.
}

\ifnum\stoc=1
	\tableoftodos
	\newpage
	\def\stocapx{2} 
\else
	\newpage
	\thispagestyle{empty}
	\tableofcontents
	\tableoftodos
	\newpage
	\def\stocapx{0} 
\fi

\def \apx{0}

\section{Introduction}
Given a sequence of independent samples $(x_1,y_1), (x_2,y_2), \dots$ from a joint distribution $P$ on $\calA \times \calB$ where Alice observes $x_1,x_2,\dots$ and Bob observes $y_1,y_2,\dots$, what is the largest correlation that they can extract if Alice applies some function to her observations and Bob applies some function to his?
The continuous version of this question -- where the extracted correlation is required to be in \emph{Gaussian} form -- was solved by Witsenhausen in $1975$ who gave (roughly) a $\poly(|\calA|,|\calB|,\log(1/\delta))$-time algorithm that estimates the best such correlation up to an additive $\delta$ \cite{witsenhausen1975sequences}. 
When the target distribution is Gaussian, the best possible correlation that is attainable is exactly the well-known ``maximal correlation coeffcient'' which was first introduced by Hirschfeld \cite{hirschfeld1935connection} and Gebelein \cite{gebelein1941statistische} and then studied by R{\'e}nyi \cite{renyi1959measures}. However, when the target distribution is not Gaussian, the best correlation is not well-understood and this is the question explored in this paper.
Specifically, we study the Boolean version of this question where the extracted correlation is required to be in the form of bits with fixed specified marginals. We give an algorithm that, given $\delta >0$, computes the best such correlation up to an additive $\delta$.

Questions such as the above are well-studied in the information theory literature under the label of ``Non-Interactive Simulation''.  The roots of this exploration go back to classical works by G\'acs and K\"orner \cite{gacs1973common} and Wyner \cite{Wyner_CommonInfo}. In this line of work, the problem is described by a source distribution $P(X,Y)$ and a target distribution $Q(U,V)$ and the goal is to determine the maximum rate at which samples of $P$ can be converted into samples of $Q$. (So the goal is to start with $n$ samples from $P$ and generate $R\cdot n$ samples from $Q$, for the largest possible $R$.) G\'acs and K\"orner considered the special case where $Q$ required the output to be a pair of identical uniformly random bits, i.e., $U = V = \Ber(1/2)$ and introduced what is now known as the \emph{G\'acs-K\"orner common information} of $P(X,Y)$ to characterize the maximum rate in terms of this quantity.  Wyner, on the other hand considered the ``inverse'' problem where $X = Y = \Ber(1/2)$ and $Q$ was arbitrary. Wyner characterized the best possible conversion rate in this setting in terms of what is now known as the \emph{Wyner common information} of $Q(U,V)$.
There is a rich history of subsequent work (see, for instance, \cite{kamath2015non} and the references within) exploring more general settings where neither $P$ nor $Q$ produces identical copies of some random variable.  In such settings, even the question of when can the rate be positive is unknown and this is the question we explore in this paper.

The Non-Interactive Simulation problem is also a generalization of the Non-Interactive Correlation Distillation problem which was studied by \cite{mossel2004coin, mossel2006non}\footnote{which considered the problem of maximizing agreement on a single bit, in various multi-party settings.}. Our setup can be thought of as a ``positive-rate'' version of the setup of G\'acs and K\"orner. Namely, for a known source distribution $P(X,Y)$, Alice and Bob are given an arbitrary number of i.i.d. samples and wish to generate \emph{one sample} from the distribution $Q(U,V)$ which is given by $U = V = \Ber(1/2)$.
(This is possible if and only if the G\'acs-K\"orner rate is positive.)

\paragraph{Motivation.}
Our motivation for studying the best discrete correlation that can be produced is twofold. On the one hand, this question forms part of the landscape of questions arising from a quest to
weaken the assumptions about randomness when it is employed in distributed
computing. Computational tasks are often solved well if parties have access to a
common source of randomness and there has been recent interest in
cryptography \cite{ahlswede1993common,ahlswede1998common,brassard1994secret,csiszar2000common,maurer1993secret, renner2005simple}, quantum computing \cite{nielsen1999conditions, chitambar2008tripartite, delgosha2014impossibility} and communication
complexity \cite{bavarian2014role,CGMS_ISR,ghazi2015communication} to study how the ability to solve these tasks gets affected
by weakening the source of randomness. In this space of investigations, it is a
very natural question to ask how well one source of randomness can be tranformed to a
different one, and Non-Interactive Simulation studies exactly this question.

On the other hand, from the analysis point of view, the Non-Interactive Simulation problem forms part of ``tensor power'' questions that have been challenging to analyze computationally. Specifically, in such questions, the quest is to understand how some quantity behaves as a function of the dimensionality of the problem as the dimension tends to infinity. Notable examples of such problems include the \emph{Shannon capacity of a graph} \cite{shannon1956zero,lovasz1979shannon} where the goal is to understand how the independence number of the power of a graph behaves as a function of the exponent.  Some more closely related examples arise in the problems of local state transformation of quantum entanglement \cite{Beigi_QuantumMaximalCorrelation, DelgoshaBeigi_QuantumHypercontractivity} and the problem of computing the entangled value of a game (see for eg, \cite{KKMTV_HardnessApprox_EntangledGames} and also the open problem \cite{openQIwiki_all_bell_inequalities}). A more recent example is the problem of computing the amortized communication complexity of a communication problem. Braverman-Rao \cite{braverman2011information} showed that this equals the information complexity of the communication problem, however the task of approximating the information complexity was only recently shown to be computable \cite{braverman2015information}. In our case, the best non-interactive simulation to get one pair of correlated bits might require many copies of $(x,y)$ drawn from $P$ and the challenge is to determine how many copies get us close.  Convergence results of this type are not obvious. Indeed, the task of approximating the Shannon capacity remains open to this day \cite{alon2006shannon}. Our work is motivated in part by the quest to understand tools that can be used to analyze such questions where rate of convergence to the desired quantity is non-trivial to bound.

\paragraph{Estimating Binary Correlations: Previous Work and our Result.}

In his work generalizing the results of G\'acs and K\"orner, Witsenhausen \cite{witsenhausen1975sequences} gave an efficient algorithm that achieves a \emph{quadratic} approximation to the Non-Interactive Simulation problem when $Q(U,V)$ is the distribution where $U$ and $V$ are marginally uniform over $\pm 1$
and $U$ is an $\rho$-correlated copy of $V$, i.e. $\Ex[UV] = \rho$ (henceforth, we refer to this distribution as $\DSBS(\rho)$).\footnote{Henceforth, we assume that bits are in the set $\{\pm 1\}$. By a \emph{quadratic} approximation, we mean an algorithm distinguishing between the cases (i) $\rho \geq 1-\eta$ and (ii) $\rho < 1-O(\sqrt{\eta})$ for any given parameter $\eta > 0$.}
Indeed, Witsenhausen introduced the Gaussian correlation problem as an intermediate step to solving this problem and his rounding
technique to convert the Gaussian random variables into Boolean ones is essentially the same as that of the Goemans-Williamson
algorithm for approximating maximum cut sizes in graphs~\cite{goemans1995maxcut}.
Already implicit from the work of Witsenhausen is that ``maximum correlation'' gives a way to upper bound the best achievable
$\rho$ when simulating $\DSBS(\rho)$. 
Recent works in the information theory community \cite{kamath2012non,kamath2015non,beigi2015duality} enhance the collection of analytical tools that can be used to show stronger impossibility results. 
While these works produce stronger bounds, they do not necessarily converge to the optimal limit and indeed basic questions about
simulation remain open. For instance, till our work, even the following question was open \cite{sudeep}: If $P$ is the uniform disribution on
$\{(0,0), (0,1), (1,0)\}$ and $Q = \DSBS(.49)$ (i.e. $U,V$ are uniformly $\pm 1$, with $\Ex[UV] = .49$), can $P$ simulate
$Q$ arbitrarily well? Our work answers such questions in principle. (Specifically we do give a finite time procedure to
approximate the best $\rho$ to within arbitrary accuracy. However, we have not run this algorithm to determine the answer to
this specific question.)

Below we state our main theorem informally (see Theorem~\ref{thm:main} for the formal
statement). 

\begin{theorem}[Informal]\label{thm:main_thm}
	There is an \emph{algorithm} that takes as inputs a source distribution $P$, a parameter $\rho >0$ and an error parameter $\delta >0$, runs in time bounded by some computable function of $P$, $\rho$ and $\delta$, and either outputs a non-interactive protocol that simulates $\DSBS(\rho)$ up to additive $\delta$ in total variation distance, or asserts that there is no protocol that gets $O(\delta)$-close to $\DSBS(\rho)$ in total variation distance.
\end{theorem}

More generally, the proof techniques extend to deciding the non-interactive simulation problem for an arbitrary $2\times 2$ target distribution. In particular, we also show the following (see Theorem~\ref{thm:main_full} for the formal statement).

\begin{theorem}[Informal]\label{thm:main_thm_full}
	There is an \emph{algorithm} that takes as inputs a source distribution $P$, a $2 \times 2$ target distribution $Q$ and an error parameter $\delta >0$, runs in time bounded by some computable function of $P$, $Q$ and $\delta$, and either outputs a non-interactive protocol that simulates $Q$ up to additive $\delta$ in total variation distance, or asserts that there is no protocol that gets $O(\delta)$-close to $Q$ in total variation distance.
\end{theorem}

The crux of Theorems~\ref{thm:main_thm} and \ref{thm:main_thm_full} is to prove \emph{computable} bounds on the number of copies of $(X,Y)$ that are needed in order to come $\delta$-close to the target distribution. We now describe the challenges towards achieving such bounds, and the
techniques we use.

\subsection{Proof Overview}

We start by describing some illustrative special cases of the problem. In the case where $P = \DSBS(\rho)$, maximal correlation based arguments imply that $\DSBS(\rho)$ is the `best' $\DSBS$ distribution that can simulated \cite{witsenhausen1975sequences}. Thus, in this case, dictators functions achieve the optimal strategy. Consider now the case where $P$ is a pair of $\rho$-correlated zero-mean unit-variance Gaussians\footnote{allowing here continuous distributions for the sake of intuition}. Then, Borell's isoperimetric inequality implies that the strategy where each of Alice and Bob outputs the sign of her/his Gaussian achieves the best possible $\DSBS$ \cite{borell1985geometric}.

Given the above two examples where a \emph{single-copy} strategy is optimal, it is tempting to try to determine the best $\DSBS$ that can be simulated using a single copy of $P$ and hope that it would be close to the optimal $\DSBS$ (i.e., to the one that can be simulated using an arbitrary number of copies of $P$). But this approach cannot work as is illustrated by the following example which shows that using many copies of $P$ is in some cases actually \emph{needed}. Consider the source joint distribution corresponding to the bipartite graph in Figure~\ref{le:alpha_ex} with $\alpha > 0$ being a small parameter (we interpret the distribution as the one obtained by sampling a random edge in the graph). This graph is the union of two components: a low-correlation component which has probability $1-\alpha$ and a perfect-correlation component which has probability $\alpha$. If we use a small number of copies of $\mu$, the corresponding samples will most likely fall in the low-correlation component, and hence the best $\DSBS$ that can be produced in such a way would have a small correlation. On the other hand, as the number of used copies becomes larger than $1/\alpha$, with high probability at least one of the corresponding samples will fall in the perfect-correlation component, and hence the resulting $\DSBS$ would have correlation very close to $1$. As another example, consider the distribution that is uniform on triples $\set{(0,0), (0,1), (1,0)}$. It follows from \cite{witsenhausen1975sequences} that it is possible to simulate $\DSBS(\third)$ using many copies of this distribution. However, it can be shown that using only a single copy of this distribution (along with private randomness), Alice and Bob can at best simulate $\DSBS(\fourth)$.

\fullver{}{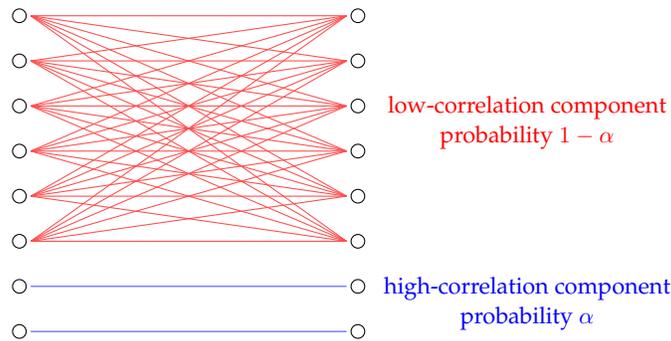
\begin{figure}
	\centering
	\begin{tikzpicture}[scale=1.5]
	\def \n {8}
	\def \nsmallplusone {7}
	\def \nsmall {6}
	\def \t {3}
	\def \radius {1.15cm}
	\def \margin {3} 
	\def \margintwo {0.07} 
	\def \radiussmall {0.7cm}
	\def \radiuslarge {1.6cm}
	\foreach \s in {1,...,\n}
	{
		\node[draw, circle, scale=0.5] (n3) at (-1.5,{(\s-1) * 0.4 -(\n-1)/2*0.4}) {};
		\node[draw, circle, scale=0.5] (n1) at (1.5,{(\s-1) * 0.40  -(\n-1)/2*0.40}) {};
		
	}
	
	\foreach \s in {1,...,\nsmall}
	{ \foreach \r in {1,...,\nsmall}
		{
			\draw[-,color=red!70] (-1.4,{(\n-1) * 0.4 -(\n-1)/2*0.4-(\s-1)*0.4}) -- (1.4,{(\n-1) * 0.4 -(\n-1)/2*0.4-(\r-1)*0.4});
			
		}
	}
	
	\foreach \s in {\nsmallplusone,...,\n}
	{
		\draw[-,color=blue!70] (-1.4,{(\n-1) * 0.4 -(\n-1)/2*0.4-(\s-1)*0.4}) -- (1.4,{(\n-1) * 0.4 -(\n-1)/2*0.4-(\s-1)*0.4});
		
	}
	
	\node[fill=none, scale=0.8, red] (n4) at (3,0.45) {\shortstack{low-correlation component \\ probability $1-\alpha$}};
	
	\node[fill=none, scale=0.8, blue] (n5) at (3,-1.15) {\shortstack{high-correlation component \\ probability $\alpha$}};
	
	\end{tikzpicture}
	\caption{Example source distribution for which many copies need to be considered.}\label{le:alpha_ex}
\end{figure}}


We now describe at a high level, the main ideas that give us the computable bound on the number of samples of the joint distribution that are sufficient to obtain a $\delta$-approximation to a given $\DSBS(\rho)$. First, we observe that the problem of deciding if one can come $\delta$-close to simulating $\DSBS(\rho)$, is equivalent to checking if Alice and Bob can non-interactively come up with a distribution $(X,Y)$ on $[-1,1] \times [-1,1]$ such that the marginals of $X$ and $Y$ have means close to $0$, but $\Ex[XY]$ is large.

The results on correlation bounds for low-influence functions (obtained using the invariance principle) \cite{mossel2005noise, mossel2010gaussianbounds}, say that if Alice and Bob are using only low-influential functions, then in fact the correlation that they get cannot be much better than that obtained by taking appropriate threshold functions on correlated gaussians. Moreover, Alice and Bob can in fact simulate correlated gaussians using only a constant number of samples from the joint distribution, by applying the maximal correlation based technique of Witsenhausen \cite{witsenhausen1975sequences}.

In the general case, we show that we can first convert Alice and Bob's functions to have {\em low degree}, after which we apply a regularity lemma (inspired from that of \cite{diakonikolas2010regularity}) to conclude that after fixing a constant number of coordinates, the restricted function is in fact low-influential. This reduces the general case to the special case of having low-influential functions and which is handled as described in the previous paragraph.

The more general case of simulating arbitrary $2 \times 2$ distribution also follows a similar outline. For a more technical overview of the proof, we refer the reader to Section~\ref{subsec:pf_overv}.

\subsection{Roadmap of the paper}

In Section~\ref{sec:prelim}, we give some of the basic definitions, etc.\fullver{ (see Appendix~\ref{apx:prelim} for a more detailed preliminaries, especially about fourier analysis and hypercontractivity)}{}. Our main theorems are also presented in this section as Theorems~\ref{thm:main_full} and \ref{thm:main}. In Section~\ref{sec:overview}, we state our main technical lemma (Theorem~\ref{thm:main-lemma}), which is used to prove Theorem~\ref{thm:main}. We also give a proof overview for Theorem~\ref{thm:main-lemma}. In Sections~\ref{sec:smoothing}, \ref{sec:reglem}, \ref{sec:correlation_bds} and \ref{sec:gaussians}, we state \fullver{and give some proof overview of}{and prove} the technical lemmas involved in proving Theorem~\ref{thm:main-lemma}. Finally, in Section~\ref{sec:together}, we put together everything to prove Theorem~\ref{thm:main-lemma}. We end with some open questions in Section~\ref{sec:discussion}.
\ifnum\apx=0
\section{Preliminaries} \label{sec:prelim}
\else
\section{Additional preliminaries} \label{apx:prelim}
\fi

\ifnum\apx=0
\subsection{Notation}
We use script letters $\calA$, $\calB$, etc. to denote finite sets, and $\mu$ will usually denote a probability distribution. $(\calA \times \calB, \mu)$ is a joint probability space. We use $\mu_A$ and $\mu_B$ to denote the marginal distributions of $\mu$. We use letters $x$, $y$, etc to denote elements of $\calA$, and bold letters $\bx$, $\by$, etc. to denote elements in $\calA^n$. We use $x_i$, $y_i$ to denote individual coordinates of $\bx$, $\by$, respectively.

For a probability space $(\calA, \mu)$, we will use the following definitions and notations borrowed from \cite{austrin2011randomly}.
\ifnum\stoc=1
	\noapx{$(\calA^n, \mu^{\otimes n})$ denotes the product space $\calA \times \calA \times \cdots \times \calA$ endowed with the product distribution. $\Supp(\mu) \defeq \setdef{x}{\mu(x) >0}$ is the support of $\mu$. We would generally assume without loss of generality that $\Supp(\mu) = \calA$. $\alpha(\mu) \defeq \min\setdef{\mu(x)}{x \in \Supp(\mu)}$ denotes the minimum non-zero probability of any atom in $\calA$ under the distribution $\mu$. $L^2(\calA,\mu)$ denotes the space of functions from $\calA$ to $\bbR$. The inner product on $L^2(\calA,\mu)$ is denoted by $\inangle{f,g}_{\mu} := \Ex\limits_{x \sim \mu}[f(x) g(x)]$. The $\ell_p$-norm by $\norm{p}{f} := \insquare{\Ex\limits_{x \sim \mu}|f(x)|^p}^{1/p}$. Also, $ \norm{\infty}{f} := \max_{\mu(x) >0} |f(x)|$. It is easy to verify that $\norm{p}{f} \le \norm{q}{f}$ for $1 \le p \le q$. For two distributions $\mu$ and $\nu$, $\dTV(\mu, \nu)$ is the total variation distance between $\mu$ and $\nu$.}{}
\else
\begin{itemize}
\item $(\calA^n, \mu^{\otimes n})$ denotes the product space $\calA \times \calA \times \cdots \times \calA$ endowed with the product distribution.
\item $\Supp(\mu) \defeq \setdef{x}{\mu(x) >0}$ is the support of $\mu$. We would generally assume without loss of generality that $\Supp(\mu) = \calA$.
\item $\alpha(\mu) \defeq \min\setdef{\mu(x)}{x \in \Supp(\mu)}$ denotes the minimum non-zero probability of any atom in $\calA$ under the distribution $\mu$.
\item $L^2(\calA,\mu)$ denotes the space of functions from $\calA$ to $\bbR$.
\item The inner product on $L^2(\calA,\mu)$ is denoted by $\inangle{f,g}_{\mu} := \Ex\limits_{x \sim \mu}[f(x) g(x)]$.
\item The $\ell_p$-norm by $\norm{p}{f} := \insquare{\Ex\limits_{x \sim \mu}|f(x)|^p}^{1/p}$. Also, $ \norm{\infty}{f} := \max_{\mu(x) >0} |f(x)|$.
\item It is easy to verify that $\norm{p}{f} \le \norm{q}{f}$ for $1 \le p \le q$.
\item For two distributions $\mu$ and $\nu$, $\dTV(\mu, \nu)$ is the total variation distance between $\mu$ and $\nu$.
\end{itemize}
\fi

\fi 

\ifnum\apx=0
\subsection{The non-interactive simulation problem}

The problem of non-interactive simulation is defined as follows,

\begin{defn}[Non-interactive simulation \cite{kamath2015non}]
	Let $(\calA \times \calB, \mu)$ and $(\calU \times \calV, \nu)$ be two probability spaces. We say that the distribution $\nu$ can be {\em non-interactively simulated} using distribution $\mu$, if there exists a sequence of functions $\set{f_n}_{n \in \bbN}$ and $\set{g_n}_{n \in \bbN}$ such that,
	$$f_n : \calA^n \to \calU \quad \quad g_n : \calB^n \to \calV$$
	and the distribution $\nu_n \sim (f_n(\bx), g_n(\by))_{\mu^{\otimes n}}$ over $\calU \times \calV$ is such that $\lim\limits_{n \to \infty}\dTV(\nu_n, \nu) = 0$.
\end{defn}

\fullver{}{\begin{figure}
\begin{center}
\begin{tikzpicture}[scale=0.8, transform shape]

\def \h{1.15}
\def \w{3}
\node[box, minimum width=2.25cm, minimum height=1.5cm, fill=black!10] (A) at (0,\h) {\Large \bf Alice};
\node[box, minimum width=2.25cm, minimum height=1.5cm, fill=black!10] (B) at (0,-\h) {\Large \bf Bob};

\node (X) at (-\w, \h)  {\Large $X^n$} edge[-latex] (A);
\node (Y) at (-\w, -\h) {\Large $Y^n$} edge[-latex] (B);

\node at (\w, \h)   {\Large $U$}   edge[latex-] (A);
\node at (\w, -\h)  {\Large $V$}   edge[latex-] (B);


\draw[snake, segment length=2mm] (X) -- (Y);

\end{tikzpicture}
\caption{Non-Interactive simulation as studied in \cite{kamath2012non, kamath2015non}}
\label{fig:non_int_sim}
\end{center}
\end{figure}
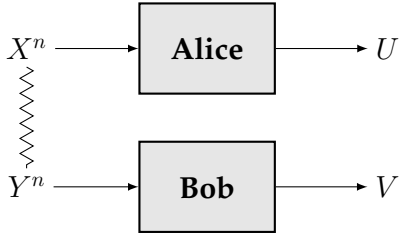}
\noindent The notion of non-interactive simulation is pictorially depicted in Figure~\ref{fig:non_int_sim}.
We formulate a natural gap-version of the non-interactive simulation problem defined as follows,

\begin{problem}[$\ANIS((\calA \times \calB, \mu), (\calU \times \calV, \nu), \delta)$] \label{prob:approx_decide_full}
	Given probability spaces $(\calA \times \calB, \mu)$ and $(\calU \times \calV, \nu)$, and an error parameter $\delta > 0$, distinguish between the following cases:
	\begin{enumerate}[(i)]
		\item there exists $N$, and functions $f : \calA^N \to \calU$ and $g : \calB^N \to \calV$, the distribution $\nu' = (f(\bx), g(\by))_{\mu^{\otimes N}}$ is such that $\dTV(\nu', \nu) \le \delta$.
		\item for all $N$ and all functions $f : \calA^N \to \calU$ and $g : \calB^N \to \calV$, the distribution $\nu' = (f(\bx), g(\by))_{\mu^{\otimes N}}$ is such that $\dTV(\nu', \nu) > 8 \delta$. \footnote{ for sake of definition, the constant $8$ could be replaced by any constant greater than $1$. For a minor technical reason however our decidability results (Theorems~\ref{thm:main_full} and \ref{thm:main}) will require this constant to be strictly greater than $2$. We choose to go ahead with $8$ for convenience.}
	\end{enumerate}
\end{problem}

The main result in this paper is the following theorem showing that the problem of $\ANIS$ is decidable when $|\calU| = |\calV| = 2$.

\begin{theorem}[Decidability of $\ANIS$ for binary targets] \label{thm:main_full}
	Given probability spaces $(\calA \times \calB, \mu)$ and $(\calU \times \calV, \nu)$ such that $|\calU| = |\calV| = 2$, and an error parameter $\delta$, there exists an algorithm that runs in time $T((\calA \times \calB, \mu), \delta)$ (which is an explicitly computable function), and decides the problem of $\ANIS((\calA \times \calB, \mu), (\calU \times \calV, \nu), \delta)$.
	
	\noindent The run time $T((\calA \times \calB, \mu), \delta)$ is upper bounded by,
	$$\exp\exp\exp\inparen{\poly\inparen{\frac{1}{\delta}, \ \frac{1}{1-\rho_0}, \ \log\inparen{\frac{1}{\alpha}}}}$$
	where $\rho_0 = \rho(\calA, \calB; \mu)$ is the maximal correlation of $(\calA \times \calB, \mu)$ (defined in Section~\ref{sec:prelim_maxcorr}) and $\alpha \defeq \alpha(\mu)$ is the minimum non-zero probability in $\mu$.
\end{theorem}

\subsubsection*{Doubly Symmetric Binary Source}

In order to ease the presentation of ideas in proving the above theorem, we restrict to a special case, where the distribution $(\calU \times \calV; \nu)$ is a {\em doubly symmetric binary source} defined below.

\begin{defn}[Doubly Symmetric Binary Source]\label{defn:DSBS}
	The distribution $\DSBS(\rho)$ is the joint distribution on $\pm 1$ random variables $(U, V)$ given by the following table,
	\begin{center}
		\begin{tabular}{c|cc}
			& $V = +1$ & $V = -1$\\
			\hline
			$U = +1$ & $(1+\rho)/4$ & $(1-\rho)/4$\\
			$U = -1$ & $(1-\rho)/4$ & $(1+\rho)/4$\\
		\end{tabular}
	\end{center}
	In particular, $\Ex[U] = \Ex[V] = 0$ and $\Ex[UV] = \rho$.
\end{defn}

We will prove a special case of Theorem~\ref{thm:main_full}, where the probability space $(\calU \times \calV, \nu)$ is the distribution $\DSBS(\rho)$ for some $\rho$ (see Theorem~\ref{thm:main} below). Even though we are proving only this special case, the main ideas involved here easily generalize to the proof of Theorem~\ref{thm:main_full}. We give a proof-sketch of this generalization in \fullver{Appendix~\ref{apx:proof_main_full}}{Section~\ref{sec:proof_main_full}}.

\begin{theorem}[Decidability of $\ANIS$ for DSBS targets] \label{thm:main}
	Given a probability space $(\calA \times \calB, \mu)$, and parameters $\rho$ and $\delta$, there exists an algorithm that runs in time $T((\calA \times \calB, \mu), \delta)$ (which is an explicitly computable function), and decides the problem of $\ANIS((\calA \times \calB, \mu), \DSBS(\rho), \delta)$.
	
	\noindent The run time $T((\calA \times \calB, \mu), \delta)$ is upper bounded by,
	$$\exp\exp\exp\inparen{\poly\inparen{\frac{1}{\delta}, \ \frac{1}{1-\rho_0}, \ \log\inparen{\frac{1}{\alpha}}}}$$
	where $\rho_0 = \rho(\calA, \calB; \mu)$ is the maximal correlation of $(\calA \times \calB, \mu)$ (defined in Section~\ref{sec:prelim_maxcorr}) and $\alpha \defeq \alpha(\mu)$ is the minimum non-zero probability in $\mu$.
\end{theorem}

\noindent We will use $\ANIS((\calA \times \calB, \mu), \rho, \delta)$ as a shorthand for $\ANIS((\calA \times \calB, \mu), \DSBS(\rho), \delta)$. Theorem~\ref{thm:main} will follow easily from the main technical lemma (Theorem~\ref{thm:main-lemma}). The proof of Theorem~\ref{thm:main}, assuming Theorem~\ref{thm:main-lemma} is present in \fullver{Appendix~\ref{apx:proof-anis}}{Section~\ref{sec:proof-anis}}.

\fi 

\subsection{Reformulation of $\ANIS$}\noapx{\label{sec:reformulation}}{\label{apx:reformulation}}

\ifnum\apx=0
With the end goal of proving Theorem~\ref{thm:main}, we introduce a new problem of Gap-Balanced-Maximum-Inner-Product, to which we show a reduction from $\ANIS$. This new formulation will be better suited for applying our techniques.

\begin{problem}[$\ABMIP((\calA \times \calB, \mu), \rho, \delta)$]\label{prob:ABMIP}
	Given a probability space $(\calA \times \calB, \mu)$, and parameters $\rho$ and $\delta$, distinguish between the following cases:
	\begin{enumerate}[(i)]
		\item there exists $N$, and functions $f : \calA^N \to [-1,1]$ and $g : \calB^N \to [-1,1]$, satisfying $|\Ex[f(\bx)]| \le \delta$ and $|\Ex[g(\by)]| \le \delta$, such that the following holds,
		$$\Ex[f(\bx) g(\by)] \ge \rho - \delta$$
		\item for all $N$ and all functions $f : \calA^N \to [-1,1]$ and $g : \calB^N \to [-1,1]$, satisfying $|\Ex[f(\bx)]| \le 2\delta$ and $|\Ex[g(\by)]| \le 2\delta$, the following holds,
		$$\Ex[f(\bx) g(\by)] < \rho - 4\delta$$
	\end{enumerate}
\end{problem}

\noindent The following proposition gives a reduction from the problem of $\ANIS$ to the problem of $\ABMIP$\fullver{ (proof present in Appendix~\ref{apx:reformulation})}{}.

\begin{proposition}\label{prop:ANIS_ABMIP_equiv}
	For any probability space $(\calA \times \calB, \mu)$ and $\rho, \delta > 0$, the following reduction holds,
	\begin{enumerate}
		\item Case (i) of $\ANIS((\calA \times \calB, \mu), \rho, \delta)$ holds $\implies$\\
		Case (i) of $\ABMIP((\calA \times \calB, \mu), \rho, 2\delta)$ holds
		\item Case (ii) of $\ANIS((\calA \times \calB, \mu), \rho, \delta)$ holds $\implies$\\
		Case (ii) of $\ABMIP((\calA \times \calB, \mu), \rho, 2\delta)$ holds
	\end{enumerate}
\end{proposition}
\fi
\ifnum\stocapx=2 \else
\ifnum\apx=1 \begin{proofof}{Proposition~\ref{prop:ANIS_ABMIP_equiv}}
\else \begin{proof}
\fi
	Both directions are relatively straight-forward.
	\begin{enumerate}
		\item If case (i) of $\ANIS((\calA \times \calB, \mu), \rho, \delta)$ holds, then there exists $N$ and functions $f : \calA^N \to \sbit$ and $g : \calB^N \to \sbit$ such that the distribution $(f(\bx), g(\by))_{\mu^{\otimes N}}$ is $\delta$-close to $\DSBS(\rho)$ in total variation distance. It follows easily from the definition of total variation distance that $|\Ex[f(\bx)]| \le 2\delta$, $|\Ex[g(\by)]| \le 2\delta$ and $\expect[f(\bx) g(\by)] \ge \rho - 2\delta$. This is exactly the conditions needed in case (i) of $\ABMIP((\calA \times \calB, \mu), \rho,2\delta)$.
		
		\item We show the contrapositive that if case (ii) of $\ABMIP((\calA \times \calB, \mu), \rho, 2\delta)$ does not hold, then in fact case (ii) of $\ANIS((\calA \times \calB, \mu), \rho, \delta)$ also does not hold. Suppose there exists $N$ and functions $f : \calA^N \to [-1,1]$ and $g : \calB^N \to [-1,1]$ such that $|\Ex[f]| \le 4\delta$, $|\Ex[g]| \le 4\delta$ and $\Ex[f(\bx) g(\by)] \ge \rho - 8\delta$. First, we observe that without loss of generality we can assume that $\Ex[f(\bx) g(\by)] \le \rho$. This is because, if that was not the case, we can suitably modify $f$ and $g$ to get $f_1 = \alpha f$ and $g_1 = \alpha g$ such that $|\Ex[f_1(\bx)]| \le 4\delta$,  $|\Ex[g_1(\by)]| \le 4\delta$ and $\Ex[f_1(\bx) g_1(\by)] = \alpha^2 \cdot \Ex[f(\bx) g(\by)]$. We can choose $\alpha$ suitably such that $\Ex [f_1(\bx) g_1(\by)] \le \rho$.
		
		To show that case (ii) of $\ANIS((\calA \times \calB, \mu), \rho, \delta)$ does not hold, we obtain {\em randomized} functions $f' : \calA^N \to \sbit$ and $g' : \calB^N \to \sbit$ as follows, $f'(\bx)$ is equal to $1$ with probability $(1+f(\bx))/2$ and $-1$ otherwise and $g'(\by)$ is equal to $1$ with probability $(1+g(\by))/2$ and $-1$ otherwise. [The randomness needed can be simulated using some additional copies of $\calA$ and $\calB$.] Note that we have the following, (i) $\Ex[f'(\bx)] = \Ex[f]$ (ii) $\Ex[g'(\by)] = \Ex[g]$ and (iii) $\rho - 8\delta \le \Ex [f'(\bx) g'(\by)] \le \rho$.
		
		Define $e_{i,j}$ for $i, j \in \sbit$ as follows,
		\begin{eqnarray*}
			e_{1,1} &=& \Pr[f'(\bx) = +1 \sAND g'(\by) = +1] - (1+\rho)/4\\
			e_{1,-1} &=& \Pr[f'(\bx) = +1 \sAND g'(\by) = -1] - (1-\rho)/4\\
			e_{-1,1} &=& \Pr[f'(\bx) = -1 \sAND g'(\by) = +1] - (1-\rho)/4\\
			e_{-1,-1} &=& \Pr[f'(\bx) = -1 \sAND g'(\by) = -1] - (1+\rho)/4
		\end{eqnarray*}
		From (i), (ii) and (iii) above, we have the following,
		\begin{eqnarray*}
			|e_{1,1} + e_{1,-1} - e_{-1,1} - e_{-1,-1}| &\le& 4\delta\\
			|e_{1,1} - e_{1,-1} + e_{-1,1} - e_{-1,-1}| &\le& 4\delta\\
			|e_{1,1} - e_{1,-1} - e_{-1,1} + e_{-1,-1}| &\le& 8\delta
		\end{eqnarray*}
		
		In addition, we have $e_{1,1} + e_{1,-1} + e_{-1,1} + e_{-1,-1} = 0$. Combining all this, it is easy to infer that $|e_{i,j}| \le 4\delta$ for any $i, j \in \sbit$. Hence for $\nu = (f(\bx), g(\by))_{\mu^{\otimes N}}$, we get that $\dTV(\nu, \DSBS(\rho)) \le 8 \delta$.
	\end{enumerate}
\ifnum\apx=1 \end{proofof}
\else \end{proof}
\fi
\fi 

\ifnum\stocapx=2
\subsection{Fourier analysis and Hypercontractivity}
We will use standard notations in Fourier analysis for functions in $L^2(\calA^n, \mu^{\otimes n})$, and use standard definitions such as Influence, Variance, etc. We will also use some concentration bounds based on hypercontractivity. Owing to space constraints, we present the requisite preliminaries in Appendices~\ref{apx:prelim_fourier} and \ref{apx:prelim_hypercontractivity}.

\else
\subsection{Fourier analysis and multi-linear polynomials} \noapx{\label{sec:prelim_fourier}}{\label{apx:prelim_fourier}}

We recall some background in Fourier analysis that will be useful to us. Let $q$ be any positive integer and let $(\calA,\mu)$ be a finite probability space with $|\calA| = q$. Let $\calX_0, \cdots, \calX_{q-1}: \calA \to \bbR$ be an orthonormal basis for the space $L^2(\calA,\mu)$ with respect to the inner product $\inangle{.,.}_{\mu}$. Furthermore, we require that this basis has the property that $\calX_0 = {\bf 1}$, i.e., the function that is identically $1$ on every element of $\calA$.

For $\bsigma = (\sigma_1, \cdots, \sigma_n) \in \bbZ_q^n$, define $\calX_{\bsigma} : \calA^n \to \bbR^n$ as follows,
$$\calX_{\bsigma}(x_1,\dots,x_n) = \prod_{i \in [n]} \calX_{\sigma_i}(x_i)$$

It is easily seen that the functions $\set{\calX_{\bsigma}}_{\bsigma \in \bbZ_q^n}$ form an orthonormal basis for the product space $L^2(\calA^n, \mu^{\otimes n})$. Thus, every function $f \in L^2(\calA^n, \mu^{\otimes n})$ can be written as
$$f(\bx) = \sum\limits_{\bsigma \in \bbZ_q^n} \what{f}(\bsigma) \calX_{\bsigma}(\bx)$$
where $\what{f}: \bbZ_q^n \to \bbR$ can be obtained as $\what{f}(\bsigma) = \inangle{f,\calX_{\bsigma}}_{\mu}$. The function $\what{f}$ is the Fourier transform of $f$ with respect to the basis $\set{\calX_i}_{i \in \bbZ_q}$. Although we will work with an arbitrary (albeit fixed) basis, many of the important properties of the Fourier transform are basis-independent. 
The most basic properties of $\what{f}$ are summarized in the following fact which follows from the orthonormality of $\set{\calX_{\bsigma}}_{\bsigma \in \bbZ_q^n}$.
\begin{fact} \label{fact:fourier_basics}
We have that:
\begin{itemize}
\item Plancherel Identity : $\Ex[f(\bx) g(\bx)] = \sum\limits_{\bsigma} \what{f}(\bsigma) \what{g}(\bsigma)$.
\item As a special case, we have Parseval's identity, $\Ex[f(\bx)^2] = \sum\limits_\bsigma \what{f}(\bsigma)^2$.
\item $\Ex[f] = \what{f}(\mathbf{0})$.
\item $\Var[f] = \sum\limits_{\bsigma \neq \mathbf{0}} \what{f}(\bsigma)^2$.
\end{itemize}
\end{fact}

In this paper, we will deal with joint probability spaces of the type $(\calA \times \calB, \mu)$. In such cases, we will denote the marginal probability spaces as $(\calA, \mu_A)$ and $(\calB, \mu_B)$. We will abuse notations, to use $\calX_{\bsigma}$ to denote the orthonormal basis vectors for both $L^2(\calA^n, \mu_A^{\otimes n})$ as well as $L^2(\calB^n, \mu_B^{\otimes n})$. The space of $\bsigma$ will be $\bbZ_{|\calA|}^n$ or $\bbZ_{|\calB|}^n$ accordingly, and will be clear from context.

For $\bsigma \in \mathbb{Z}_q^n$, the {\em degree} of $\bsigma$ is denoted by $\inabs{\bsigma} \defeq \inabs{\setdef{i \in [n]}{\sigma_i \neq 0}}$. We say that the degree of a function\footnote{we will interchangeably use the word {\em polynomial} to talk about any function in $L^2(\calA^n, \mu^{\otimes n})$.} $f \in L^2(\calA^n, \mu^{\otimes n})$, denoted by $\deg(f)$, is the largest value of $|\bsigma|$ such that $\what{f}(\bsigma) \ne 0$.

\begin{definition}[Influence]
For every coordinate $i \in [n]$, $\Inf_i(f)$ is the $i$-th influence of $f$, and $\Inf(f)$ is the total influence, which are defined as
$$\Inf_i(f) \defeq \Ex_{\bx_{-i}} \insquare{\Var_{x_i} \, [f(\bx)]} \quad \quad \quad \Inf(f) \defeq \sum\limits_{i=1}^n \Inf_i(f)$$
\end{definition}
\noindent The basic properties of influence are summarized in the following fact.

\begin{fact}\label{fact:influences}
For any function $f \in L^2(\calA^n, \mu^{\otimes n})$, we have the following:
\begin{enumerate}
\item[(i)] $\Inf_i(f) = \sum\limits_{\bsigma : \sigma_i \ne \mathbf{0}} \what{f}(\bsigma)^2$ and hence, for all $i$, $\Inf_i(f) \le \Var(f)$
\item[(ii)] $\Inf(f) = \sum\limits_{\bsigma} |\bsigma| \cdot \what{f}(\bsigma)^2$
\item[(iii)] If $\deg(f) = d$, then $\Inf(f) \le d \cdot \Var[f]$.
\end{enumerate}
\end{fact}

\subsubsection*{Restrictions of polynomials}

We will often use restrictions of polynomials. For any subset $H \subseteq [n]$, we will use $\bx_H$ to denote the tuple of variables in $\bx$ with indices in $H$. For any function $P \in L^2(\calA^n, \mu^{\otimes n})$, and any $\xi \in \calA^{H}$, we will use $P_{\xi}$ to denote the function obtained by restriction of $\bx_H$ to $\xi$, that is, $P_{\xi}(\bx_{T}) = P(\bx_H \gets \xi, \bx_T)$ (where $T = [n] \setminus H$); whenever we use such terminology, the subset $H$ will be clear from context.	We will use the phrase ``$\xi$ fixes $H$ over $\calA$'' to mean such a restriction. We will use $\set{\bsigma_H}$ to denote all degree sequences in $\bbZ_q^H$, and similarly $\set{\bsigma_T}$ to denote all degree sequences in $\bbZ_q^T$. We use $\bsigma_H \circ \bsigma_T$ to denote $\bsigma \in \bbZ_q^n$ such that $\sigma_i = (\sigma_H)_i \text{ if } i \in H \text{ or } (\sigma_T)_i \text{ if } i \in T$.

\noindent We now state a lemma that will be needed,

\begin{lem}[cf. Lemma 3.3 in \cite{diakonikolas2010regularity}] \label{lem:exp_inf}
For any function $P \in L^2(\calA^n, \mu^{\otimes n})$, consider a random assignment $\xi \sim \mu_A^H$ to the variables $\bx_H$. Let $T = [n] \setminus H$. Then, for all $i \in T$, it holds that $\bbE_{\xi} [\Inf_i (P_{\xi})] = \Inf_i(P)$. Also, $\bbE_{\xi} [\Var(P_{\xi})] \le \Var(P)$.
\end{lem}

To prove the lemma, we first recall the following fact about the expected value of Fourier coefficients under random restrictions.
\begin{fact}\label{fact:expected_four_inf}
Let $P \in L^2(\calA^n, \mu^{\otimes n})$. For any subset $H \subseteq [n]$, consider an assignment $\xi$ to the variables $\bx_H$. Let $T = [n] \setminus H$. Then, we have
$$\what{P}_{\xi}(\bsigma_T) = \sum\limits_{\bsigma_H} \what{P}(\bsigma_H \circ \bsigma_T) \cdot \calX_{\bsigma_H}(\xi)$$ and therefore
$$\Ex_{\xi}\insquare{\what{P}_{\xi}(\bsigma_T)^2} = \sum\limits_{\bsigma_H} \what{P}(\bsigma_H \circ \bsigma_T)^2 $$
\end{fact}
\begin{proof}
The first part follows from simply substituting $P_{\xi}(\bx_T) = P(\bx_H \gets \xi, \bx_T)$, and taking inner product with $\calX_{\bsigma_T}(\bx_T)$.
$$\what{P}_{\xi}(\bsigma_T) = \inangle{\sum\limits_{\bsigma_H \circ \bsigma_T'} \what{P}(\bsigma_H \circ \bsigma_T') \cdot \calX_{\bsigma_H}(\xi) \calX_{\bsigma_T'}(\bx_T)\ \ , \ \calX_{\bsigma_T}(\bx_T)}_{\mu} = \sum\limits_{\bsigma_H} \what{P}(\bsigma_H \circ \bsigma_T) \cdot \calX_{\bsigma_H}(\xi)$$

The second part simply follows from the orthonormality of the characters $\calX_{\bsigma_H}$ and $\calX_{\bsigma_H'}$ for $\bsigma_H \ne \bsigma_H'$. In particular, we have the following,
\begin{eqnarray*}
\Ex_{\xi}\insquare{\what{P}_{\xi}(\bsigma_T)^2}
&=& \Ex_{\xi}\insquare{\inparen{\sum\limits_{\bsigma_H} \what{P}(\bsigma_H \circ \bsigma_T) \cdot \calX_{\bsigma_H}(\xi)}^2}\\
&=& \Ex_{\xi}\insquare{\sum\limits_{\bsigma_H \bsigma_H'} \what{P}(\bsigma_H \circ \bsigma_T) \cdot \what{P}(\bsigma_H' \circ \bsigma_T) \cdot \calX_{\bsigma_H}(\xi) \cdot \calX_{\bsigma_H'}(\xi)}\\
&=& \sum\limits_{\bsigma_H \bsigma_H'} \what{P}(\bsigma_H \circ \bsigma_T) \cdot \what{P}(\bsigma_H' \circ \bsigma_T) \cdot \Ex_{\xi}\insquare{\calX_{\bsigma_H}(\xi) \cdot \calX_{\bsigma_H'}(\xi)}\\
&=& \sum\limits_{\bsigma_H} \what{P}(\bsigma_H \circ \bsigma_T)^2
\end{eqnarray*}

\end{proof}

Intuitively, the above fact says that all the Fourier weight on degree sequences $\set{\bsigma_H \circ \bsigma_T}_{\bsigma_H}$ {\em collapses} down onto $\bsigma_T$ in expectation. Consequently, the influence of an unrestricted variable does not change, and the variance does not increase in expectation under random restrictions, as both these quantities are sums of Fourier weight on certain $\bsigma_T$'s.\\

\begin{proofof}{Lemma~\ref{lem:exp_inf}}
We simply use Facts~\ref{fact:fourier_basics} and \ref{fact:influences} in addition to Fact \ref{fact:expected_four_inf} to prove the lemma.

\noindent Basically, from Facts~\ref{fact:fourier_basics} and \ref{fact:expected_four_inf} we get,
$$\Ex_{\xi} [\Var(P_{\xi})]
= \Ex_{\xi} \insquare{\sum_{\bsigma_T \ne \mathbf{0}} \what{P}_{\xi}(\bsigma_T)^2}
= \sum_{\bsigma_T \ne \mathbf{0}} \Ex_{\xi} \insquare{\what{P}_{\xi}(\bsigma_T)^2}
= \sum_{\bsigma_T \ne \mathbf{0}} \sum_{\bsigma_H} \what{P}(\bsigma_H \circ \bsigma_T)^2
\le \Var(P)$$

\noindent Similarly, from Facts~\ref{fact:influences} and \ref{fact:expected_four_inf} we get that for all $i \in T$,
$$\Ex_{\xi} [\Inf_i(P_{\xi})]
= \Ex_{\xi} \insquare{\sum_{\substack{\bsigma_T : \\ (\sigma_T)_i \ne 0}} \what{P}_{\xi}(\bsigma_T)^2}
= \sum_{\substack{\bsigma_T : \\ (\sigma_T)_i \ne 0}} \Ex_{\xi} \insquare{\what{P}_{\xi}(\bsigma_T)^2}
= \sum_{\substack{\bsigma_T : \\ (\sigma_T)_i \ne 0}} \sum_{\bsigma_H} \what{P}(\bsigma_H \circ \bsigma_T)^2
= \Inf_i(P)$$
\end{proofof}


\subsection{Hypercontractivity and moment bounds}
\noapx{\label{sec:prelim_hypercontractivity}}{\label{apx:prelim_hypercontractivity}}

The following moment bound for low-degree polynomials appears as Theorem $2.7$ in \cite{austrin2011randomly}, which in turn follows from hypercontractivity.
\begin{theorem}[\cite{wolff2007hypercontractivity}] \label{thm:mom_bd_hyperc}
Let $(\calA,\mu)$ be a finite probability space in which the minimum non-zero probability is $\alpha(\mu) \le \half$. Then, for $p \ge 2$, every degree-$d$ polynomial $f \in L^2(\calA^n, \mu^{\otimes n})$ satisfies
$$\norm{p}{f} \le C_p(\alpha)^{d/2} \norm{2}{f}$$
Here, $C_p$ is defined by
$$ C_p(\alpha) = \frac{A^{1/p'}-A^{-1/p'}}{A^{1/p}-A^{-1/p}} $$
where $A = (1-\alpha)/\alpha$ and $1/p + 1/p' = 1$. The value at $\alpha = 1/2$ is taken to be the limit of the above expression as $\alpha \to 1/2$, i.e., $C_p(1/2) = p-1$.
\end{theorem}

\noindent We will use the following known concentration bound for low-degree polynomials.
\begin{theorem}[\cite{austrin2011randomly}; Theorem $2.12$] \label{thm:conc_bd}
Let $f \in L^2(\calA^n, \mu^{\otimes n})$ be a degree-$d$ polynomial. Then, for any $t > e^{d/2}$,
$$\Pr[|f| > t \cdot \norm{2}{f}] \le \exp(-c t^{2/d})$$
where $c:= \frac{\alpha(\mu) d}{e}$.
\end{theorem}

\begin{defn}[Bonami-Beckner operator]
For any $\rho \in [0,1]$, the Bonami-Beckner operator $T_{\rho}$ on a probability space $(\calA, \mu)$ is given by its action on any $f : \calA \to \bbR$, as follows,
$$(T_\rho f)(x) = \Ex[f(Y) | X = x]$$
where the conditional distribution of \ $Y$ given $X = x$ is $\rho \delta_x + (1-\rho) \mu$ where $\delta_x$ is the delta measure on $x$. In other words, given $X = x$, $Y$ is obtained by either setting it to $x$ with probability $\rho$ or independently sampling from $\mu$ with probability $(1-\rho)$.

For the product space $(\calA^n, \mu^{\otimes n})$, we define the Bonami-Beckner operator $T_{\rho}$ as, $T_{\rho} = \otimes_{i=1}^n T_{\rho}^{(i)}$, where $T_{\rho}^{(i)}$ is the Bonami-Beckner operator on the $i$-th coordinate $(\calA, \mu)$.
\end{defn}
\fi

\subsection{Maximal Correlation and Witsenhausen's rounding}
\noapx{\label{sec:prelim_maxcorr}}{\label{apx:prelim_maxcorr}}

\ifnum\apx=0
The ``maximal correlation coeffcient'' was first introduced by Hirschfeld \cite{hirschfeld1935connection} and Gebelein \cite{gebelein1941statistische} and then studied by R{\'e}nyi \cite{renyi1959measures}.

\begin{defn}[Maximal correlation] \label{def:max_corr}
Given a joint probability space $(\calA \times \calB, \mu)$, we define the maximal correlation of the joint distribution $\rho(\calA, \calB; \mu)$ as follows,
$$\rho(\calA, \calB; \mu) = \sup \set{\Ex\limits_{(x, y) \sim \mu} [f(x) g(y)] \ \bigg | \ \inmat{f : \calA \to \bbR, & \Ex[f] = \Ex[g]= 0\\ g : \calB \to \bbR, & \Var(f) = \Var(g) = 1}}$$
\end{defn}
\fi

\ifnum\stocapx=2 \else
\noindent Maximal correlation has the following properties which imply necessary conditions for when non-interactive simulation could be possible.

\begin{fact}[Properties of maximal correlation (cf. \cite{bryc2005maximumcorrelation})] \mbox{}
	
\begin{enumerate}
	\item (Tensorization) : For all joint probability spaces $(\calA \times \calB, \mu)$, it is the case that $\rho(\calA^n, \calB^n; \mu^{\otimes n}) = \rho(\calA, \calB; \mu)$.
	\item (Data processing) : For all joint probability spaces $(\calA \times \calB, \mu)$, and any functions $f : \calA \to \calU$ and $g : \calB \to \calV$, it is the case that $\rho(\calA, \calB; \mu) \ge \rho(\calU, \calV; \nu)$, where $\nu$ is the distribution $(f(x), g(y))_{(x,y) \sim \mu}$.
	\item (Lower Semi-Continuous) : If distributions $(\calU \times \calV; \nu_n)$ are such that $\lim_{n \to \infty} \nu_n = \nu$, then $\lim_{n \to \infty} \rho(\calU, \calV; \nu_n) \ge \rho(\calU, \calV; \nu)$.
\end{enumerate}
\end{fact}

\begin{cor}[Necessary condition for non-interactive simulation]
Let $(\calA \times \calB, \mu)$ and $(\calU \times \calV, \nu)$ be two probability spaces. If the distribution $\nu$ can be {\em non-interactively simulated} using distribution $\mu$, then it must be the case that $\rho(\calA, \calB; \mu) \ge \rho(\calU, \calV; \nu)$.
\end{cor}
\fi

\ifnum\stocapx=2
Maximal correlation has certain properties which imply necessary conditions for when non-interactive simulation is possible (See Appendix~\ref{apx:prelim_maxcorr} for more details). In addition, using a result of Witsenhausen~\cite{witsenhausen1975sequences}, we have the following theorem,
\else
A simple fact that can be easily verified is that the maximal correlation of the distribution $\DSBS(\rho)$ is $\rho$. And hence if $(\calA \times \calB, \mu)$ can non-interactively simulate $\DSBS(\rho^*)$, then $\rho^* \le \rho(\calA, \calB; \mu)$. In addition, Witsenhausen\cite{witsenhausen1975sequences} showed that any joint probability space $(\calA \times \calB, \mu)$ can simulate $\DSBS(\rho^*)$ for $\rho^* = 1 - \frac{2 \arccos(\rho(\calA, \calB; \mu))}{\pi}$. \fullver{Putting this together we obtain Theorem~\ref{thm:witsenhausen}.}{All together, we have the following theorem,}
\fi

\ifnum\apx=0
\begin{theorem}[Witsenhausen \cite{witsenhausen1975sequences}] \label{thm:witsenhausen}
For any joint probability space $(\calA \times \calB, \mu)$, with $\rho = \rho(\calA, \calB; \mu)$, then the largest $\rho^*$ for which $(\calA \times \calB, \mu)$ can non-interactively simulate $\DSBS(\rho^*)$ is bounded as follows,
$$1 - \frac{2 \arccos(\rho)}{\pi} \quad \le \quad \rho^* \quad \le \quad \rho$$
\end{theorem}

Note that, maximal correlation is an easily computable quantity, namely, it is the second largest singular value of the Markov operator\footnote{The Markov operator corresponding to $(\calA \times \calB, \mu)$ is a $|\calA| \times |\calB|$ matrix $T$ which is given by $T(x,y) = \mu(y | X = x)$.} corresponding to $(\calA \times \calB, \mu)$.

\begin{remark}
The astute reader might have noticed a strong resemblance between Theorem~\ref{thm:witsenhausen} and the random hyperplane rounding of Goemans-Williamson \cite{goemans1995maxcut} used in the approximation algorithm for MAX-CUT. This is not a coincidence and indeed the bounds in Theorem~\ref{thm:witsenhausen} come from morally the same technique as in \cite{goemans1995maxcut}.
\end{remark}

\noindent In this context we will use the following shorthand for $\rho$-correlated 2-dimensional gaussian.

\begin{definition}[2-dimensional Gaussian] \label{defn:2d_gauss}
Let $\calG(\rho)$ denote a 2-dimensional gaussian distribution with mean $\inbmat{0 \\ 0}$ and covariance matrix $\inbmat{1 & \rho \\ \rho & 1}$.
\end{definition}
\fi

\ifnum\stocapx=2 \else
\subsection{2-dimensional Berry-Esseen theorem}

We will need the following 2-dimensional Berry-Esseen theorem. The proof is very similar to Theorem 68 of \cite{matulef2010testing}. The main difference is that in our case the random variables are not necessarily binary-valued, but they do have finite support. We include the proof for completeness.

\begin{lem}[2-dimensional Berry-Esseen] \label{lem:2d_CLT}
	Let $(X,Y)$ be any pair of correlated real-valued random variables with finite support such that, $\Ex[X] = \Ex[Y] = 0$ and $\Var(X) = \Var(Y) = 1$ and $\Ex[XY] = \rho$. For every $\zeta > 0$, there exists $w \defeq w((X,Y), \zeta) \in \bbN$, such that for every $a, b \in \bbR$, it is the case that,
	$$\inabs{\ \Pr[\overline{X} \le a, \overline{Y} \le b] \ - \ \Pr[\calG_1 \le a, \calG_2 \le b] \ } \le \zeta$$
	where $\overline{X} = \frac{\sum_{i=1}^w X_i}{\sqrt{w}}$, $\overline{Y} = \frac{\sum_{i=1}^w Y_i}{\sqrt{w}}$ (with $(X_i, Y_i)$ draw i.i.d. from $(X,Y)$) and $(\calG_1, \calG_2) \sim \calG(\rho)$.
	
	In particular, one may take $w = O\inparen{\frac{1+\rho}{\alpha \cdot (1-\rho)^{3} \cdot \zeta^2}}$, where $\alpha$ is the minimum non-zero probability in the distribution $(X,Y)$.
\end{lem}

In order to prove Lemma~\ref{lem:2d_CLT}, we need the following statement that appears as Theorem 16 in \cite{khot2007optimal} and as Corollary 16.3 in \cite{bhattacharya1986normal}.

\begin{theorem}\label{thm:gen_Berry_Esseen}
	Let $Z_1, \dots, Z_w$ be independent random variables taking values in $\mathbb{R}^k$ and satisfying:
	\begin{itemize}
		\item $\Ex[Z_j]$ is the all-zero vector for every $j \in \{1,\dots,w\}$. 
		\item $\sum_{j=1}^w \Cov[Z_j]/w = V$ where $\Cov$ denotes the covariance matrix.
		\item $\lambda$ is the smallest eigenvalue of $V$ and $\Lambda$ is the largest eigenvalue of $V$.
		\item $\sum_{j=1}^w \Ex\insquare{\norm{}{Z_j}^3}/w = \rho_3 < \infty$.
	\end{itemize}
	Let $Q_w$ denote the distribution of $(Z_1+\dots+Z_w)/\sqrt{w}$, let $\Phi_{0,V}$ denote the distribution of the $k$-dimensional Gaussian with mean $0$ and covariance matrix $V$, and let $\eta = C \lambda^{-3/2} \rho_3 w^{-1/2}$, where $C$ is a certain universal constant. Then, for any Borel set $A$,
	$$\inabs{Q_w(A) - \Phi_{0,V}(A)} \le \eta + B(A)$$
	where $B(A)$ is the following measure of the boundary of $A$: $B(A) = 2\sup_{y \in \mathbb{R}^k} \Phi_{0,V}((\partial A)^{\eta'} + y)$, $\eta' = \Lambda^{1/2} \eta$ and $(\partial A)^{\eta'}$ denotes the set of points within distance $\eta'$ of the topological boundary of $A$.
\end{theorem}

\begin{proofof}{Lemma~\ref{lem:2d_CLT}}
	We apply Theorem~\ref{thm:gen_Berry_Esseen} with $k=2$. Let $Z = (X, Y)$, and hence we have that,
	$$\Ex[Z] = \inbmat{0 \\ 0} \quad \sAND \quad \Cov[Z] = \inbmat{1 & \rho \\ \rho & 1}$$
	Let $Z_i = (X_i, Y_i)$. Since all $Z_i$ are i.i.d. distributed according to $Z$, we have $V = \sum_{j=1}^w \Cov[(X_j,Y_j)]/w$ is also $\inbmat{1 & \rho \\ \rho & 1}$. It follows that the smallest and largest eigenvalues of $V$ are $\lambda = 1-\rho$ and $\Lambda = 1+\rho$ respectively. Moreover, since the underlying distribution has finite support, we have that,
	$$\sum_{j=1}^w \frac{\Ex\insquare{\norm{}{Z_j}^3}}{w} = \Ex\insquare{\norm{}{Z}^3} < \max \norm{}{Z} \cdot \Ex\insquare{\norm{}{Z}^2} \le 1/\sqrt{\alpha}$$ (where $\alpha$ is the smallest atom in the distribution $(X,Y)$). Thus, we get $\rho_{3} \le 1/\sqrt{\alpha}$. Hence, $\eta = O((1-\rho)^{-3/2}\alpha^{-1/2} w^{-1/2})$. As in \cite{khot2007optimal}, one can check that the topological boundary of any set of the form $(-\infty,a] \times (-\infty,b]$ is $O(\eta')$, where $\eta' = (1+\rho)^{1/2} \eta$. Thus, from Lemma~\ref{lem:2d_CLT} it follows by choosing $w$ sufficiently large so that $O\big((1+(1+\rho)^{1/2})(1-\rho)^{-3/2} \alpha^{-1/2} w^{-1/2}\big) \le \zeta$.
	
	In particular it suffices to choose $w = O\inparen{\frac{(1+(1+\rho)^{1/2})^2}{\alpha \cdot (1-\rho)^{3} \cdot \zeta^2}} = O\inparen{\frac{1+\rho}{\alpha \cdot (1-\rho)^{3} \cdot \zeta^2}}$.
\end{proofof}
\fi
\ifnum\apx=0
\section{Main Technical Lemma and Overview} \label{sec:overview}

In this section we state the main technical lemma which will be used to solve $\ABMIP$. We also give a high level overview of the proof techniques.

\begin{theorem}\label{thm:main-lemma}
Given any joint probability space $(\calA \times \calB, \mu)$ and any $\delta > 0$, there exists $n_0 = n_0((\calA \times \calB, \mu), \delta)$ such that for any $n$ and any functions $f : \calA^n \to [-1,1]$ and $g : \calB^n \to [-1,1]$, there exist functions $\wtilde{f} : \calA^{n_0} \to [-1,1]$ and $\wtilde{g} : \calB^{n_0} \to [-1,1]$ such that $\inabs{\Ex[\wtilde{f}] - \Ex[f]} \le \delta/3$, $\inabs{\Ex[\wtilde{g}] - \Ex[g]} \le \delta/3$ and
$$\Ex_{(\bx, \by) \sim \mu^{\otimes n_0}} \insquare{\wtilde{f}(\bx) \cdot \wtilde{g}(\by)} \quad \ge \quad \Ex_{(\bx, \by) \sim \mu^{\otimes n}} \insquare{f(\bx) \cdot g(\by)} - \delta$$

\noindent Most importantly, $n_0$ is a computable function in the parameters of the problem. In particular, one may take,
$$n_0 = \exp\inparen{\poly\inparen{\frac{1}{\delta}, \ \frac{1}{1-\rho}, \ \log\inparen{\frac{1}{\alpha}}}}$$
where $\rho \defeq \rho(\calA, \calB; \mu)$ is the maximal correlation of $(\calA \times \calB, \mu)$ and $\alpha \defeq \alpha(\mu)$ is the minimum non-zero probability in $\mu$.
\end{theorem}

\subsection{Proof overview}\label{subsec:pf_overv}

The proof of Theorem~\ref{thm:main-lemma} goes through a series of intermediate steps, which we describe at a high level here. At each step we lose only a small amount in the correlation. The first three steps preserve the marginals $\bbE [f]$ and $\bbE[g]$ exactly, while the fourth step incurs a small additive error in the same. The full proof is presented in Section~\ref{sec:together}.

\begin{itemize}
\item [{\bf Step 1:}] {\em Smoothing of strategies.} We transform $f$ and $g$ into functions $f_1$, $g_1$ such that $f_1$ and $g_1$ have `most' of their Fourier mass concentrated on terms of degree at most $d$, where $d$ is a constant that depends on the distribution $(\calA \times \calB, \mu)$ and a tolerance parameter, but is independent of $n$. This transformation is described in Section~\ref{sec:smoothing}.

\item [{\bf Step 2:}] {\em Regularity lemma for low degree functions.} We first prove a {\em regularity lemma} (similar to the one in \cite{diakonikolas2010regularity}) which roughly shows that for any degree-$d$ polynomial, there exists a $h$-sized subset of variables, such that under a random restriction of the variables in this subset, the resulting function on the remaining variables has low individual influences (i.e. $\le \tau$). Note that $h$ will be a constant depending on the degree $d$ and $\tau$, but will be independent of $n$.

We apply this regularity lemma on the degree-$d$ truncated versions of both $f_1$ and $g_1$ obtained from Step 1. We take the union of the subsets obtained for $f_1$ and $g_1$. We show that with high probability over random restrictions of the variables in this subset, the resulting restriction of $f_1$ and $g_1$ on the remaining variables has low individual influences. This step is described in Section~\ref{sec:reglem}.

Note that this step does not change the functions $f_1$ and $g_1$ at all, but we gain some structural knowledge about the same.

\item [{\bf Step 3:}] {\em Correlation bounds for low influence functions.} We use results about correlation bounds for low influential functions \cite{mossel2005noise, mossel2010gaussianbounds}. Intuitively, these results suggest that if the functions $f_1$ and $g_1$ were low influential functions to begin with, then the correlation $\bbE[f_1(\bx) g_1(\by)]$ will not be `much' better than the correlation between certain threshold functions applied on correlated gaussians.

We apply the above correlation bounds for the low influential functions obtained by restrictions of the small subset of variables in $f_1$ and $g_1$, to obtain functions $f_2 : \calA^h \times \bbR \to [-1,1]$ and $g_2 : \calB^h \times \bbR \to [-1,1]$, where Alice and Bob together have access to $h$ samples from $(\calA \times \calB, \mu)$ and a single copy of $\rho$-correlated gaussians, that is, $\calG(\rho)$ (see Defn. \ref{defn:2d_gauss}). Here the correlation $\rho$ is same as the maximal correlation $\rho(\calA, \calB; \mu)$. This step is described in Section~\ref{sec:correlation_bds}.

\item [{\bf Step 4:}] {\em Simulating correlated gaussians.} Finally, Alice and Bob can non-interactively simulate the distribution $\calG(\rho)$ using constantly many samples from $(\calA \times \calB, \mu)$. This is done using the technique of Witsenhausen \cite{witsenhausen1975sequences}, which primarily uses a 2-dimensional central limit theorem. This step is described in Section~\ref{sec:gaussians}.

\end{itemize}
\fi  

\ifnum\stocapx=2 \else
\ifnum\stoc=1
\section{Decidability of $\ANIS$} \label{apx:proof-anis}
\else
\subsection{Decidability of $\ANIS$} \label{sec:proof-anis}
\fi

Assuming Theorem~\ref{thm:main-lemma}, we now give the algorithm as described in Theorem~\ref{thm:main}.\\

\begin{proofof}{Theorem~\ref{thm:main}}
We have from Proposition~\ref{prop:ANIS_ABMIP_equiv} that, in order to decide $\ANIS((\calA \times \calB, \mu), \rho, \delta)$, it suffices to decide $\ABMIP((\calA \times \calB, \mu), \rho, 2\delta)$.

If we were in the YES case of $\ABMIP((\calA \times \calB, \mu), \rho, 2\delta)$, then we have that there exists an $n$ and functions $f : \calA^n \to [-1,1]$ and $g : \calB^n \to [-1,1]$, such that $|\Ex[f(\bx)]| \le 2\delta$, $|\Ex[g(\by)]| \le 2\delta$ and $\Ex[f(\bx) \cdot g(\by)] \ge \rho - 2\delta$. Using Theorem~\ref{thm:main-lemma}, with parameter $\delta$, we get that there exists functions $\wtilde{f} : \calA^{n_0} \to [-1,1]$ and $\wtilde{g} : \calB^{n_0} \to [-1,1]$ such that $\inabs{\Ex[\wtilde{f}(\bx)]} \le 8\delta/3$, $\inabs{\Ex[\wtilde{g}(\by)]} \le 8\delta/3$ and $\Ex[\wtilde{f}(\bx) \cdot \wtilde{g}(\by)] \ge \rho - 3\delta$.

In the NO case of $\ABMIP((\calA \times \calB, \mu), \rho, 2\delta)$, we have that for all $n$, in particular for $n = n_0$, there do not exist functions $f : \calA^n \to [-1,1]$ and $g : \calB^n \to [-1,1]$ such that $|\Ex[f(\bx)]| \le 4\delta$, $|\Ex[g(\by)]| \le 4\delta$ and $\Ex[f(\bx) \cdot g(\by)] \ge \rho - 8\delta$.
 
This naturally gives us a brute force algorithm: Analyze all possible functions $\wtilde{f} : \calA^{n_0} \to [-1,1]$ and $\wtilde{g} : \calB^{n_0} \to [-1,1]$ to check if there exist functions satisfying $\inabs{\Ex[\wtilde{f}(\bx)]} \le 8\delta/3$, $\inabs{\Ex[\wtilde{g}(\by)]} \le 8\delta/3$ and $\Ex[\wtilde{f}(\bx) \cdot \wtilde{g}(\by)] \ge \rho - 3\delta$. For purposes of our algorithm we can treat the range $[-1,1]$ as a discrete set $R \defeq \setdef{k\delta^2/10}{k \in \bbZ, |k| < 10/\delta^2}$. This ensures that if indeed such a desired $\wtilde{f}$ and $\wtilde{g}$ exist, then we will find functions $\wtilde{f}' : \calA^{n_0} \to R$ and $\wtilde{g}' : \calB^{n_0} \to R$ such that $\inabs{\Ex[\wtilde{f}'(\bx)]}, \inabs{\Ex[\wtilde{g}'(\by)]} \le 8\delta/3 + O(\delta^2)$ and $\Ex[\wtilde{f}'(\bx) \cdot \wtilde{g}'(\by)] \ge \rho - 3\delta- O(\delta^2)$. In the YES case, we will find such functions, whereas in the NO case, $\wtilde{f}'$ and $\wtilde{g}'$ as above simply don't exist. 

It is easy to see that this brute force can be done in $\inparen{\frac{1}{\delta^2}}^{O\inparen{(|\calA| \cdot |\calB|)^{n_0}}}$ time, which is upper bounded by the running time claimed in Theorem~\ref{thm:main}.
\end{proofof}
\fi 
\section{Smoothing of Strategies}
\noapx{\label{sec:smoothing}}{\label{apx:smoothing}}

\ifnum\apx=0
The first step in our approach is to obtain smoothed versions of the functions $f : \calA^n \to [-1,1]$ and $g : \calB^n \to [-1,1]$, which have {\em small Fourier tails}, without hurting the correlation by much. In particular, we show the following lemma\fullver{ (proof in Appendix~\ref{apx:smoothing})}{}.

\begin{lem}[Smoothing of strategies] \label{lem:smoothing}
Given any joint probability space $(\calA \times \calB, \mu)$ and parameters $\lambda, \eta > 0$, there exists $d = d((\calA \times \calB, \mu), \lambda, \eta)$ such that for any $n$ and any functions $f : \calA^n \to [-1,1]$ and $g : \calB^n \to [-1,1]$, there exist functions $f_1 : \calA^n \to [-1,1]$ and $g_1 : \calB^n \to [-1,1]$ such that $\Ex[f_1] = \Ex[f]$ and $\Ex[g_1] = \Ex[g]$, and
$$\inabs{\Ex_{(\bx, \by) \sim \mu^{\otimes n}} \insquare{f_1(\bx) \cdot g_1(\by)} - \Ex_{(\bx, \by) \sim \mu^{\otimes n}} \insquare{f(\bx) \cdot g(\by)}} \quad \le \quad \lambda$$
such that $f_1$ and $g_1$ have low energy Fourier tails, namely,
$$\sum_{|\bsigma| > d} \what{f}_1(\bsigma)^2 \le \eta \quad \sAND \quad \sum_{|\bsigma| > d} \what{g}_1(\bsigma)^2 \le \eta$$

\noindent In particular, one may take $d = \frac{\log \eta}{2 \log \gamma}$, where $\gamma = 1 - C \frac{(1-\rho)\lambda}{\log(1/\lambda)}$, and $\rho = \rho(\calA, \calB; \mu)$.
\end{lem}
\fi
\ifnum\stocapx=2 \else
\noindent To prove Lemma~\ref{lem:smoothing}, we use Lemma 6.1 of Mossel \cite{mossel2010gaussianbounds}. We state a specialized version of Mossel's lemma, which suffices for our application.

\begin{lem}[\cite{mossel2010gaussianbounds}] \label{lem:mossel_smoothing}
Let $(\calA \times \calB, \mu)$ be finite joint probability space, such that $\rho(\calA \times \calB, \mu) \le \rho$.

\noindent Let $P \in L^2(\calA^n, \mu_A^{\otimes n})$ and $Q \in L^2(\calB^n, \mu_B^{\otimes n})$ be multi-linear polynomials. Let $\eps > 0$ and $\gamma$ be chosen sufficiently close to $1$ so that,
$$\gamma \ge (1-\eps)^{\log \rho/(\log \eps + \log \rho)}$$
Then:
$$\inabs{\expect[P(\bx)Q(\by)] - \expect[T_{\gamma}P(\bx) T_{\gamma}Q(\by)]} \le 2\eps \Var[P] \Var[Q]$$
In particular, there exists an absolute constant $C$ such that it suffices to take
$$\gamma \defeq 1 - C \frac{(1-\rho)\eps}{\log (1/\eps)}$$
\end{lem}

\begin{proofof}{Lemma~\ref{lem:smoothing}}
Given parameters $\lambda$ and $\eta$, we first choose $\eps$ and $\gamma$ in Lemma~\ref{lem:mossel_smoothing}, such that $\eps = \lambda/2$ and $\gamma = 1 - C \inparen{(1-\rho)\eps}/\inparen{\log (1/\eps)}$ as required. We choose $d$ to be large enough such that $\gamma^{2d} \le \eta$, that is, $d = (\log \eta)/(2 \log \gamma)$. Now, given functions $f : \calA^n \to [-1, 1]$ and $g : \calB^n \to [-1,1]$, we obtain functions $f_1$ and $g_1$ as follows: $f_1(\bx) = T_{\gamma}f(\bx)$ and $g_1(\by) = T_{\gamma}g(\by)$. It is easy to see that, $\Ex[f_1(\bx)] = \Ex[f(\bx)]$ and $\Ex[g_1(\by)] = \Ex[g(\by)]$. From Lemma~\ref{lem:mossel_smoothing}, and the fact that $\Var[f], \Var[g] \le 1$, we get $\inabs{\Ex[f_1(\bx) g_1(\by)] - \Ex[f(\bx) g(\by)]} \le 2\eps = \lambda$ as desired. Also, note that $\what{f_1}(\bsigma) = \what{f}(\bsigma)\cdot\gamma^{|\bsigma|}$ (similarly for $\what{g_1}(\bsigma)$). Thus, we get that,
\begin{eqnarray*}
\sum_{|\bsigma| > d} \what{f}_1(\bsigma)^2 \quad \le \quad \gamma^{2d} \cdot \sum_{|\bsigma| > d} \what{f}(\bsigma)^2 \quad \le \quad \gamma^{2d} \quad \le \quad \eta\\
\sum_{|\bsigma| > d} \what{g}_1(\bsigma)^2 \quad \le \quad \gamma^{2d} \cdot \sum_{|\bsigma| > d} \what{g}(\bsigma)^2 \quad \le \quad \gamma^{2d} \quad \le \quad \eta
\end{eqnarray*}
\end{proofof}
\fi 
\section{Joint Regularity Lemma for Fourier Concentrated Functions}
\noapx{\label{sec:reglem}}{\label{apx:reglem}}

\ifnum\apx=1
This section contains the full proof of Lemma~\ref{lem:joint_reg}.
\else
The second step in our approach is to apply a {\em regularity lemma} on the functions $f_1 : \calA^n \to [-1,1]$ and $g_1 : \calB^n \to [-1,1]$ obtained from the previous step of smoothing. {\em Regularity lemma} is a loosely referred term which shows that for various types of combinatorial objects, an arbitrary object can be approximately decomposed into a constant number of ``pseudorandom'' sub-objects.

Our version of the regularity lemma draws inspiration from that of \cite{diakonikolas2010regularity}; in fact our proofs also closely follow theirs. Formally, we show the following lemma\fullver{ (proof in Appendix~\ref{apx:reglem})}{}.

\begin{lem}[Joint regularity lemma for Fourier-concentrated functions] \label{lem:joint_reg}
Let $(\calA \times \calB, \mu)$ be a joint probability space. Let $d \in \bbN$ and $\tau > 0$ be any given constant parameters. There exists an $\eta \defeq \eta(\tau) > 0$ and $h \defeq h((\calA \times \calB, \mu), d, \tau)$ such that the following holds:

For all $P \in L^2(\calA^n, \mu_{A}^{\otimes n})$ and $Q \in L^2(\calB^n, \mu_{B}^{\otimes n})$ satisfying $\sum_{|\bsigma| > d} \what{P}(\bsigma)^2 \le \eta$, $\sum_{|\bsigma| > d} \what{Q}(\bsigma)^2 \le \eta$, and $\Var[P] \le 1$ and $\Var[Q] \le 1$: there exists a subset of indices $H \subseteq [n]$ with $|H| \le h$, such that the restrictions of the functions $P$ and $Q$ obtained by evaluating the coordinates in $H$ according to distribution $\mu$, satisfy the following (where we denote $T = [n] \setminus H$),
\begin{itemize}
\item With probability at least $1-\tau$ over $\xi \sim \mu_{A}^{\otimes h}$, the restriction $P_{\xi}(\bx_T)$ is such that for all $i \in T$, it is the case that $\Inf_i(P_{\xi}(\bx_T)) \le \tau$ 
\item With probability at least $1-\tau$ over $\xi \sim \mu_{B}^{\otimes h}$, the restriction $Q_{\xi}(\bx_T)$ is such that for all $i \in T$, it is the case that $\Inf_i(Q_{\xi}(\bx_T)) \le \tau$ 
\end{itemize}
\noindent In particular, one may take $\eta = \tau^2/16$ and $h = \frac{d}{\tau^2} \cdot \inparen{\frac{C_4(\alpha)}{\alpha} \log \frac{C_4(\alpha)}{\alpha \cdot d \cdot \tau}}^{O(d)}$ which is a constant that depends on $d$, $\tau$ and $\alpha \defeq \alpha(\mu)$, which is the minimum non-zero probability in $\mu$.\fullver{ See Appendix~\ref{apx:prelim_hypercontractivity} for the definition of $C_4(\alpha)$, which is the hypercontractivity parameter.}{}
\end{lem}
\fi

\ifnum\stocapx=2 \else
\subsection{Regularity Lemma for Constant Degree Polynomials}

We first prove a version of the above regularity lemma for degree-$d$ functions, as opposed to Fourier-concentrated functions.

\begin{lem}[Regularity Lemma for degree-$d$ functions]\label{lem:our_reg_lem}
Let $(\calA, \mu_A)$ be a probability space. Let $d \in \bbN$ and $\tau > 0$ be any given constant parameters. There exists $h \defeq h((\calA, \mu_A), d, \tau)$ such that the following holds:

For all degree-$d$ multilinear polynomials $P \in L^2(\calA^n, \mu_{A}^{\otimes n})$ with $\Var[P] \le 1$, there exists a subset of indices $H_0 \subseteq [n]$ with $|H_0| \le h$, such that for any superset $H \supseteq H_0$, the restrictions of $P$ obtained by evaluating the coordinates in $H$ according to distribution $\mu_A$, satisfies the following (where we denote $T = [n] \setminus H$):
$$\Pr\limits_{\xi \sim \mu_A^{\otimes |H|}} \insquare{\forall i \in T : \Inf_i(P_{\xi}(\bx_T)) \le \tau} \ge 1-\tau$$
In other words, with probability at least $1-\tau$ over the random restriction $\xi \sim \mu_{A}^{\otimes |H|}$, the restricted function $P_{\xi}(\bx_T)$ is such that $\Inf_i(P_{\xi}(\bx_T)) \le \tau$ for all $i \in T$.

\noindent In particular, one may take $h = \frac{d}{\tau} \cdot \inparen{\frac{C_4(\alpha)}{\alpha} \log \frac{C_4(\alpha)}{\alpha \cdot d \cdot \tau}}^{O(d)}$ which is a constant that depends on $d$, $\tau$ and $\alpha \defeq \alpha(\mu_A)$.
\end{lem}

The intuitive explanation of the regularity lemma is as follows: If $P$ is a degree $d$ polynomial with $\Var(P) \le 1$, then the total influence of $P$ is at most $d$. Hence for all $\beta > 0$, there can only be at most $h \defeq d/\beta$ variables with influence greater than $\beta$. Indeed, our subset $H_0$ will essentially be the set of all the variables with influence at least $\beta$ (we will choose $\beta$ to be suitably smaller than $\tau$, but with no dependence on $n$). Clearly, $|H_0| \le h$. For any superset $H \supseteq H_0$, and for a random restriction of $\bx_H$ to $\xi$, it will follow from well known hypercontractivity bounds (Theorem~\ref{thm:conc_bd}) and a careful union bound, that the influence of all the remaining variables will be less than $\tau$ with high probability.
\fi

\ifnum\apx=0
Our regularity lemma draws inspiration from the one in \cite{diakonikolas2010regularity}. In fact, our proof of the above regularity lemma also closely follows the proof steps in \cite{diakonikolas2010regularity}. However their regularity lemma was much more involved as they were dealing with low-degree polynomial threshold functions, whereas we are directly dealing with low-degree polynomials. In particular, a major difference in our regularity lemmas is that \cite{diakonikolas2010regularity} obtain a (potentially) {\em adaptive} decision tree, whereas we obtain just a single subset $H$. Also, our notion of `regularity' is much simpler in that we only need all influences to be small. Another aspect of our regularity lemma is that it is robust enough to also work for Fourier concentrated functions, as opposed to only low-degree functions (potentially, \cite{diakonikolas2010regularity} could also be modified to have this feature, although it was not required for their application). Another minor difference is that our Fourier analysis is for functions in $L^2(\calA^n, \mu_A^{\otimes n})$, as opposed to functions on the boolean hypercube. But this is not really a significant difference and the proof steps go through as it is, albeit with slightly different parameters which depend on the hypercontractivity parameters of the distribution $(\calA, \mu_A)$.\\
\fi

\ifnum\stocapx=2 \else
\noindent Before we give a proof of Lemma~\ref{lem:our_reg_lem}, we would need the following claim.

\begin{claim}[cf. Claim 3.12 in \cite{diakonikolas2010regularity}]\label{clm:indiv_inf_bd}
Let $P \in L^2(\calA^n, \mu_{A}^{\otimes n})$ be a degree-$d$ polynomial. Let $H \subseteq [n]$ and $T = [n] \setminus H$. Let $\xi$ be a random restriction fixing $H$. For all $r \ge e^{d}$ and all $i \in T$, we have the following,
$$\Pr\limits_{\xi} \ [\Inf_i(P_{\xi}) > r \cdot C_4(\alpha)^{d} \cdot \Inf_i(P)] \le \exp(-c \cdot r^{1/d})$$
	
\noindent where, $c = \alpha(\mu_A) d / e$ (see Theorem~\ref{thm:conc_bd}) and $C_4(\alpha)$ is obtained as in Theorem~\ref{thm:mom_bd_hyperc}.
\end{claim}
\begin{proof}
The identity $\Inf_{i}(P_{\xi}) = \sum_{\bsigma_T : (\sigma_T)_i \ne 0} \what{P}_{\xi}(\bsigma_T)^2$ and Fact~\ref{fact:expected_four_inf} imply that $\Inf_{i}(P_{\xi})$ is a degree-$2d$ polynomial in $\xi$. Hence, the claim would follow from the concentration bound for low-degree polynomials, i.e., Theorem~\ref{thm:conc_bd}, if we can appropriately upper bound the $\ell_2$-norm of the polynomial $\Inf_{i}(P_{\xi})$. So, to prove Claim~\ref{clm:indiv_inf_bd}, it suffices to show that
\begin{equation}\label{eq:inf_l2_bd}
\norm{2}{\Inf_{i}(P_{\xi})} \le C_4(\alpha)^d \cdot \Inf_{i}(P)
\end{equation}
\myignore{Let $Q(\xi) \defeq \Inf_i(P_{\xi})$.} By the triangle inequality for norms we have that,
\begin{equation*}
\norm{2}{\Inf_{i}(P_{\xi})} = \norm{2}{\sum\limits_{\bsigma_T : (\sigma_T)_i \ne 0} \what{P}_{\xi}(\bsigma_T)^2} \le \sum\limits_{\bsigma_T : (\sigma_T)_i \ne 0} \norm{2}{\what{P}_{\xi}(\bsigma_T)^2}
\end{equation*}
Since $\what{P}_{\xi}(\sigma_T)$ is a degree-$d$ polynomial, the moment bound for low-degree polynomials, i.e., Theorem~\ref{thm:mom_bd_hyperc}, yields that
$$ \norm{2}{\what{P}_{\xi}(\bsigma_T)^2} = \norm{4}{\what{P}_{\xi}(\bsigma_T)}^2 \le C_4(\alpha)^d \norm{2}{\what{P}_{\xi}(\bsigma_T)}^2$$
and hence
\begin{eqnarray*}
\norm{2}{\Inf_{i}(P_{\xi})} & \le & C_4(\alpha)^d \sum\limits_{\bsigma_T : (\sigma_T)_i \ne 0} \norm{2}{\what{P}_{\xi}(\bsigma_T)}^2 \\
& = & C_4(\alpha)^d \sum\limits_{\bsigma_T : (\sigma_T)_i \ne 0} \Ex_{\xi} \insquare{\what{P}_{\xi}(\bsigma_T)^2} \\
& = & C_4(\alpha)^d \cdot \Ex_{\xi} \insquare{\Inf(P_{\xi})} \\
& = & C_4(\alpha)^d \cdot \Inf_i(P)
\end{eqnarray*}
where the last equality follows from Lemma~\ref{lem:exp_inf}. Thus, Equation (\ref{eq:inf_l2_bd}) and the claim follows from Theorem~\ref{thm:conc_bd}.
\end{proof}

\begin{proofof}{Lemma~\ref{lem:our_reg_lem}}
Let $P \in L^2(\calA^n, \mu_{A}^{\otimes n})$ be the given degree-$d$ multilinear polynomial with $\Var[P] \le 1$. From part (iii) of Fact~\ref{fact:influences}, we have that $\Inf(P) \le d$. Let $H_0 \subset [n]$ be the set of indices $i \in [n]$ such that $\Inf_i(f) \ge \beta$. Since, $d \ge \Inf(P) \ge \sum_i \Inf_i(P)$, we have that $|H_0| \le d/\beta$. We will choose $\beta$ as a suitable constant less than $\tau$, but with no dependence on $n$.

Fix $H \supseteq H_0$ and let $T = [n] \setminus H$. From Claim~\ref{clm:indiv_inf_bd}, we have for any $i \in T$, that $\Pr_{\xi} \ [\Inf_i(P_{\xi}) > r \cdot C_4(\alpha)^{d} \cdot \Inf_i(P)] \le \exp(-\Omega(c \cdot r^{1/d}))$. However, to prove that $\Inf_i(P_{\xi}) \le \tau$ for all $i \in T$, with high probability, we cannot simply use a \naive union bound over all $i \in T$, as that will introduce a dependence of $n$ in $\beta$ and thereby in $h$. Instead, we use a {\em bucketing} argument, as done in \cite{diakonikolas2010regularity}, as follows:

We partition the indices $i \in T$ into buckets $\set{B_j}_{j \in \bbN}$ as $B_j = \setdef{i \in T}{\Inf_i(P) \in \left (\frac{\beta}{2^{j+1}}, \frac{\beta}{2^{j}} \right ]}$. Since $\Inf(P) \le d$, we have that $|B_j| \le 2^{j+1} d/\beta$. For all $i \in B_j$, we use the concentration $\Pr\limits_{\xi} \ [\Inf_i(P_{\xi}) \le r \cdot C_4(\alpha)^{d} \cdot \Inf_i(P)] \ge 1 - \exp(-c \cdot r^{1/d})$ by choosing $r = \frac{\tau \cdot 2^j}{\beta \cdot C_4(\alpha)^d}$. We then do a union bound over all the buckets. Thus, we get that,

\begin{eqnarray*}
\Pr\limits_{\xi} \insquare{\forall i \in T : \Inf_i(P_{\xi}(\bx_T)) \le \tau} &\ge& 1 - \sum_{j=0}^{\infty} \Pr\limits_{\xi} \insquare{\exists i \in B_j : \Inf_i(P_{\xi}(\bx_T)) > \tau}\\
&\ge& 1 - \sum_{j=0}^{\infty} \exp\inparen{-c \inparen{\frac{\tau \cdot 2^j}{\beta \cdot C_4(\alpha)^d}}^{1/d}} \cdot \frac{2^{j+1}d}{\beta}
\end{eqnarray*}

\noindent It can be verified that for $\frac{1}{\beta} = \frac{(2 \cdot C_4(\alpha))^d}{c^d \cdot \tau} \cdot \log \inparen{\frac{(2 \cdot C_4(\alpha))^d}{c^d \cdot \tau}}^{d}$ it holds that,
$$\sum_{j=0}^{\infty} \exp\inparen{-c \inparen{\frac{\tau \cdot 2^j}{\beta \cdot C_4(\alpha)^d}}^{1/d}} \cdot \frac{2^{j+1}d}{\beta} \quad \le \quad \tau$$

\noindent Thus, we have the regularity lemma as desired with $|H_0| \le h = \frac{d}{\beta} = \frac{d}{\tau} \cdot \inparen{\frac{C_4(\alpha)}{\alpha} \log \frac{C_4(\alpha)}{\alpha \cdot d \cdot \tau}}^{O(d)}$ which is a constant that depends on $d$, $\tau$ and $\alpha \defeq \alpha(\mu_A)$.

\end{proofof}

\subsection{Joint Regularity Lemma}

In this section, we use Lemma~\ref{lem:our_reg_lem} to prove the joint regularity lemma, namely Lemma~\ref{lem:joint_reg}.

\begin{proofof}{Lemma~\ref{lem:joint_reg}}
We have $P \in L^2(\calA^n, \mu_{A}^{\otimes n})$ and $Q \in L^2(\calB^n, \mu_{B}^{\otimes n})$ satisfying $\sum_{|\bsigma| > d} \what{P}(\bsigma)^2 \le \eta$, $\sum_{|\bsigma| > d} \what{Q}(\bsigma)^2 \le \eta$, and $\Var[P] \le 1$ and $\Var[Q] \le 1$. First, we split $P$ and $Q$ into low and high degree components. That is, $P(\bx) = P^{\ell}(\bx) + P^{h}(\bx)$ and $Q(\by) = Q^{\ell}(\by) + Q^{h}(\by)$, where $P^{\ell}(\bx)$ and $Q^{\ell}(\by)$ contain all the monomials of degree at most $d$ in $P(\bx)$ and $Q(\by)$ respectively. Note that $\Var[P^\ell] \le \Var[P] \le 1$. Similarly, $\Var[Q^\ell] \le 1$.

We apply the regularity lemma for degree-$d$ functions (Lemma~\ref{lem:our_reg_lem}), with parameter $\tau$ equal to $\tau/4$, on functions $P^{\ell}$ and $Q^{\ell}$ separately, to obtain subsets $H_A, H_B \subseteq [n]$ respectively. The subset $H$ is then obtained as $H_A \cup H_B$. Note that, $|H| \le h((\calA, \mu_A), d, \tau/4) + h((\calB, \mu_B), d, \tau/4)$, which is a computable in terms of the parameters of the problem, but more importantly has no dependence on $n$.

\noindent From Lemma~\ref{lem:our_reg_lem}, we know that for $T = [n] \setminus H$ (note that $H \supseteq H_A$ and $H \supseteq H_B$),
\begin{eqnarray}
\Pr\limits_{\quad \quad \xi \sim \mu_A^{\otimes |H|}} \quad \insquare{\forall i \in T : \Inf_i(P^{\ell}_{\xi}(\bx_{T})) \le \tau/4} &\ge& 1-\tau/4 \label{eqn:jreg_bd1}\\
\Pr\limits_{\quad \quad \xi \sim \mu_B^{\otimes |H|}} \quad \insquare{\forall i \in T : \Inf_i(Q^{\ell}_{\xi}(\by_{T})) \le \tau/4} &\ge& 1-\tau/4 \label{eqn:low_deg_inf}
\end{eqnarray}

\noindent Now, we show that after adding $P^h$ to $P^{\ell}$, the influences $\Inf_i(P_{\xi}(\bx_{T}))$ are still upper bounded by $\tau$, with high probability over $\xi$.
\begin{eqnarray}
\Inf_i(P_{\xi}(\bx_{T})) &=& \sum\limits_{\bsigma_T : (\sigma_T)_i \ne 0} \inparen{\sum_{\bsigma_H} \what{P}(\bsigma_H \circ \bsigma_T) \cdot \chi_{\bsigma_H}(\xi)}^2 \nonumber\\
&=& \sum\limits_{\bsigma_T : (\sigma_T)_i \ne 0} \inparen{\sum_{\bsigma_H} \what{P^{\ell}}(\bsigma_H \circ \bsigma_T) \cdot \chi_{\bsigma_H}(\xi) + \sum_{\bsigma_H} \what{P^h}(\bsigma_H \circ \bsigma_T) \cdot \chi_{\bsigma_H}(\xi)}^2 \nonumber\\
&\le& 2 \cdot \sum\limits_{\bsigma_T : (\sigma_T)_i \ne 0} \inparen{\sum_{\bsigma_H} \what{P^{\ell}}(\bsigma_H \circ \bsigma_T) \cdot \chi_{\bsigma_H}(\xi)}^2 + \inparen{\sum_{\bsigma_H} \what{P^h}(\bsigma_H \circ \bsigma_T) \cdot \chi_{\bsigma_H}(\xi)}^2 \nonumber\\
&=& 2 \cdot \inparen{\Inf_i(P_{\xi}^{\ell}(\bx_T)) + \Inf_i(P_{\xi}^{h}(\bx_T))} \label{eqn:inf_tr_in}
\end{eqnarray}

\noindent Since $\bbE_{\xi} \insquare{\Var(P_{\xi}^{h}(\bx_T))} \le \Var(P^{h}(\bx_T)) \le \eta \ $ (see Lemma~\ref{lem:exp_inf}), we have by Markov's inequality that,

$$\Pr\limits_{\quad \quad \xi \sim \mu_A^{\otimes |H|}} \quad \insquare{\Var(P^{h}_{\xi}(\bx_{T})) \le 4\eta/\tau} \quad \ge \quad 1-\tau/4$$

\noindent Since for all $i \in T$, we have $\Inf_i(P_{\xi}^{h}(\bx_T)) \le \Var(P_{\xi}^{h}(\bx_T))$ (see Fact~\ref{fact:influences}), we get that

\begin{equation}
\Pr\limits_{\quad \quad \xi \sim \mu_A^{\otimes |H|}} \quad \insquare{\forall i \in T : \Inf_i(P^{h}_{\xi}(\bx_{T})) \le 4\eta/\tau} \quad \ge \quad 1-\tau/4 \label{eqn:high_deg_inf}
\end{equation}

\noindent We will choose $\eta = (\tau/4)^2$, and thus, by union bound (using Equations~\ref{eqn:inf_tr_in}, \ref{eqn:low_deg_inf} and \ref{eqn:high_deg_inf}), we have that, 

$$\Pr\limits_{\quad \quad \xi \sim \mu_A^{\otimes |H|}} \quad \insquare{\forall i \in T : \Inf_i(P_{\xi}(\bx_{T})) \le \tau} \quad \ge \quad 1-\tau/2 \quad > \quad 1 - \tau$$




\noindent By exactly same flow of calculations for $Q(\by)$, we can have,
$$\Pr\limits_{\quad \quad \xi \sim \mu_B^{\otimes |H|}} \quad \insquare{\forall i \in T : \Inf_i(Q_{\xi}(\by_{T})) \le \tau} \quad \ge \quad 1-\tau/2 \quad > \quad 1 - \tau$$

\noindent This completes the proof of Lemma~\ref{lem:joint_reg}.

\end{proofof}
\fi
\section{Applying correlation bounds for low-influence functions}
\noapx{\label{sec:correlation_bds}}{\label{apx:correlation_bds}}

\ifnum\apx=1
This section contains the full proof of Lemma~\ref{lem:app_corr_bds}.
\else
The third step in our approach is to use {\em correlation bounds for low-influence functions} obtained from the invariance principle \cite{mossel2005noise, mossel2010gaussianbounds}, to convert the functions $f_1 : \calA^n \to [-1,1]$ and $g_1 : \calB^n \to [-1,1]$ into functions $f_2 : \calA^h \times \bbR \to [-1,1]$ and $g_2 : \calB^h \times \bbR \to [-1,1]$ using the following lemma\fullver{ (proof in Appendix~\ref{apx:correlation_bds})}{}.

\begin{lem}[Applying correlation bounds for low-influence functions] \label{lem:app_corr_bds}
Let $(\calA \times \calB, \mu)$ be a joint probability space. Let $\gamma > 0$ be any given constant parameter. There exists a $\tau \defeq \tau((\calA \times \calB, \mu), \gamma) > 0$ such that the following holds:

For all functions $f_1 : \calA^n \to [-1,1]$ and $g_1 : \calB^n \to [-1,1]$, and a subset $H \subseteq [n]$ with $|H| = h$, such that the restrictions of the functions $f_1$ and $g_1$ obtained by evaluating the coordinates in $H$ according to distribution $\mu$, satisfy the following (where we denote $T = [n] \setminus H$),
\begin{itemize}
\item With probability at least $1-\tau$ over $\xi \sim \mu_{A}^{\otimes h}$, the restriction $(f_1)_{\xi}(\bx_T)$ is such that for all $i \in T$, it is the case that $\Inf_i((f_1)_{\xi}(\bx_T)) \le \tau$
\item With probability at least $1-\tau$ over $\xi \sim \mu_{B}^{\otimes h}$, the restriction $(g_1)_{\xi}(\bx_T)$ is such that for all $i \in T$, it is the case that $\Inf_i((g_1)_{\xi}(\bx_T)) \le \tau$
\end{itemize}

\noindent There exist functions $f_2 : \calA^h \times \bbR \to [-1,1]$ and $g_2 : \calB^h \times \bbR \to [-1,1]$, such that,
$$\expect_{\bx \sim \mu_A^{\otimes n}} f_1(\bx) = \expect_{\substack{\bx \sim \mu_A^{\otimes h} \\ r_A \sim \calN(0,1)}} f_2(\bx, r_A) \quad \sAND \quad \expect_{\by \sim \mu_B^{\otimes n}} g_1(\by) = \expect_{\substack{\by \sim \mu_B^{\otimes h} \\ r_B \sim \calN(0,1)}} g_2(\by, r_B)$$

\noindent and,

$$\expect_{\substack{(\bx,\by)\sim \mu^{\otimes h} \\ (r_A, r_B) \sim \calG(\rho)}} \insquare{f_2(\bx, r_A) \cdot g_2(\by, r_B)} \quad \ge \quad \expect_{(\bx,\by) \sim \mu^{\otimes n}} \insquare{f_1(\bx) \cdot g_1(\by)} - \gamma $$

Additionally, $f_2$ and $g_2$ will have the following special form: there exist functions $f_2' : \calA^h \to \bbR$ and $g_2' : \calB^h \to \bbR$ such that,

$$f_2(\bx, r) = \infork{1 & r \ge f_2'(\bx) \\ -1 & r < f_2'(\bx)} \quad \sAND \quad g_2(\by, r) = \infork{1 & r \ge g_2'(\by) \\ -1 & r < g_2'(\by)}$$

\noindent Also, one may take $\tau = \gamma^{O\inparen{\frac{\log(1/\gamma) \log(1/\alpha)}{(1-\rho)\gamma}}}$, where $\rho = \rho(\calA, \calB; \mu)$ and $\alpha \defeq \alpha(\mu)$ is the minimum non-zero probability in $\mu$.
\end{lem}
\fi
\ifnum\stocapx=2
The main technical tool in proving Lemma~\ref{lem:app_corr_bds} is a result about correlation bounds for low influence functions (which are generalizations of the `Majority is Stablest' theorem), which is obtained from the invariance principle \cite{mossel2005noise, mossel2010gaussianbounds}.
\else
As mentioned before, the main technical tool in proving Lemma~\ref{lem:app_corr_bds} is a result about correlation bounds for low influence functions (which are generalizations of the `Majority is Stablest' theorem). Before we state that theorem, we need the following definition, which is a slightly modified version of Definition 1.12 in \cite{mossel2010gaussianbounds}.

\begin{defn}[Gaussian stability] \label{defn:gaussian_stability}
Let $\Phi$ be the cumulative distribution function (CDF) of a standard $\mathcal{N}(0,1)$ Gaussian. Given $\rho \in [-1,1]$ and $\mu, \nu \in [-1,1]$, we define,
\begin{eqnarray*}
\overline{\Gamma}_{\rho}(\mu,\nu) &=& \Ex [\overline{P}_{\mu}(X) \cdot \overline{Q}_{\nu}(Y)]\\
\underline{\Gamma}_{\rho}(\mu,\nu) &=& - \Ex [\overline{P}_{\mu}(X) \cdot \overline{Q}_{-\nu}(Y)]
\end{eqnarray*}
where $(X,Y)$ is distributed according to $\calG(\rho)$ and 
$$\overline{P}_{\mu}(X) = \infork{1 & X \le \Phi^{-1}(\frac{1+\mu}{2}) \\ -1 & \text{otherwise}} \quad \sAND \quad \overline{Q}_{\nu}(X) = \infork{1 & Y \le \Phi^{-1}(\frac{1+\nu}{2}) \\ -1 & \text{otherwise}}$$
\myignore{
\noindent \unsure{we don't need this definition. Should we still define it?}Similarly, we define,
$$\underline{\Gamma}_{\rho}(\mu,\nu) = \Ex [\underline{P}_{\mu}(X) \cdot \underline{Q}_{\nu}(Y)]$$
where $(X,Y)$ is distributed according to $\calG(\rho)$ and 
$$\underline{P}_{\mu}(X) = \infork{1 & X \le \Phi^{-1}(\frac{1+\mu}{2}) \\ -1 & \text{otherwise}} \quad \sAND \quad \underline{Q}_{\nu}(X) = \infork{1 & Y \ge \Phi^{-1}(\frac{1-\nu}{2}) \\ -1 & \text{otherwise}}$$
}

\noindent Note that for $(X, Y) \sim \calG(\rho)$, we have that,
$$\Ex_{X} \insquare{\overline{P}_{\mu}(X)} = \mu \quad \sAND \quad \Ex_{Y} \insquare{\overline{Q}_{\nu}(Y)} = \nu = \Ex_{Y} \insquare{-\overline{Q}_{-\nu}(Y)}$$
\end{defn}

\noindent With this definition in hand, we can state the correlation bounds for low influential functions that are obtained from invariance principle.

\begin{theorem}[Correlation bounds from invariance principle; \cite{mossel2005noise, mossel2010gaussianbounds}]\label{th:inv_princ}
Let $(\calA \times \calB, \mu)$ be a joint probability space. As before, let $\alpha = \alpha(\mu)$ be the minimum probability of any atom in $\calA \times \calB$. Let $\rho = \rho(\calA, \calB; \mu)$ be the maximal correlation of the joint probability space (see Definition~\ref{def:max_corr}).
	
Then, for all $\eps > 0$, there exists $\tau \defeq \tau((\calA \times \calB, \mu), \eps) > 0$ such that if
$$P : \calA^n \to [-1,1] \quad \sAND \quad Q : \calB^n \to [-1,1]$$
satisfy $\Inf_i(P) \le \tau$ and $\Inf_i(Q) \le \tau$ for all $i \in [n]$, then
$$\underline{\Gamma}_{\rho}\inparen{\Ex_{\bx}[P(\bx)] \ , \ \Ex_{\by}[Q(\by)]} \ - \ \eps
\quad \le \quad \Ex_{(\bx, \by) \sim \mu^{\otimes n}}\insquare{P(\bx) Q(\by)} \quad \le \quad
\overline{\Gamma}_{\rho}\inparen{\Ex_{\bx}[P(\bx)] \ , \ \Ex_{\by}[Q(\by)]} \ + \ \eps$$
Furthermore, one may take
$$\tau = \eps^{O\inparen{\frac{\log(1/\eps) \log(1/\alpha)}{(1-\rho)\eps}}}$$
\end{theorem}

\noindent Intuitively, this theorem says that if $P$ and $Q$ are low-influential, then their correlation is not much more than that of appropriate threshold functions applied on $\rho$-correlated gaussians. With this tool in hand, we are now ready to prove Lemma~\ref{lem:app_corr_bds}.\\

\begin{proofof}{Lemma~\ref{lem:app_corr_bds}}
Suppose we have $f_1 : \calA^n \to [-1,1]$ and $g_1 : \calB^n \to [-1,1]$, and a subset $H \subseteq [n]$ with $|H| = h$, such that the restrictions of the functions $f_1$ and $g_1$ obtained by evaluating the coordinates in $H$ according to distribution $\mu$, satisfy the properties as stated in the lemma. We construct function $f_2 : \calA^h \times \bbR \to [-1,1]$ and $g_2 : \calB^h \times \bbR \to [-1,1]$ by replacing the functions obtained after restricting the variables in $H$ by appropriate threshold functions acting on $\rho$-correlated gaussians, namely,

$$\forall \ (\bx,r) \in \calA^h \times \bbR \quad : \quad f_2(\bx, r) = \overline{P}_{\nu_1}(r) \quad \text{where} \quad \nu_1 \ \defeq \  \expect\limits_{\bx_T \sim \mu_A^{\otimes n-h}} \ \insquare{f_1(\bx_H \gets \bx, \bx_T)}$$
$$\forall \ (\by,r) \in \calB^h \times \bbR \quad : \quad g_2(\by, r) = \overline{Q}_{\nu_2}(r) \quad \text{where} \quad \nu_2 \ \defeq \ \expect\limits_{\by_T \sim \mu_B^{\otimes n-h}} \ \insquare{f_1(\by_H \gets \by, \by_T)}$$

\noindent where $\overline{P}_{\nu}$ and $\overline{Q}_{\nu}$ are as defined in Definition~\ref{defn:gaussian_stability}.\footnote{For simplicity, we will abuse notations in the followng sense: when we say $f_1(\bx)$, we mean $\bx \in \calA^n$, but when we say $f_2(\bx,r)$, we mean $\bx \in \calA^h$ and $r \in \bbR$.}\\

\noindent It follows from definition, that $\Ex [f_2(\bx,r)] = \Ex [f_1(\bx)]$ and $\Ex [g_2(\by,r)] = \Ex [g_1(\by)]$. That is, this process has not changed the individual means of $f_1$ and $g_1$. We now need to prove that the correlation is not hurt by much. From Lemma~\ref{lem:joint_reg} and a simple union bound, we know that with probability $1-2\tau$, a random restriction $(\bx_H, \by_H)$ for the coordinates in $H$ is such that,

$$\forall i \in T \quad : \Inf_i((f_1)_{\bx_H}(\bx_T)) \le \tau \quad \sAND \quad \Inf_i((g_1)_{\by_H}(\by_T)) \le \tau$$

\noindent Let's call all the tuples $(\bx_H, \by_H)$ for which the above happens as `good'.

\begin{eqnarray*}
& &\hspace*{-7mm}\Ex_{\bx, \by} \ f_1(\bx) g_1(\by) \\
&=& \Ex_{\bx_H, \by_H} \insquare{\Ex_{\bx_T, \by_T} \ f_1(\bx_H, \bx_T) \cdot g_1(\by_H, \by_T)}\\
&=& \Pr[(\bx_H, \by_H) \text{ is not `good'}] \cdot \Ex_{\bx_H, \by_H} \insquare{\Ex_{\bx_T, \by_T} \ f_1(\bx_H, \bx_T) \cdot g_1(\by_H, \by_T) \bigg | (\bx_H, \by_H) \text{ is not `good'}}\\
& & + \ \Pr[(\bx_H, \by_H) \text{ is `good'}] \cdot \Ex_{\bx_H, \by_H} \insquare{\Ex_{\bx_T, \by_T} \ f_1(\bx_H, \bx_T) \cdot g_1(\by_H, \by_T) \bigg | (\bx_H, \by_H) \text{ is `good'}}\\
&\le& \Pr[(\bx_H, \by_H) \text{ is not `good'}] \cdot 1\\
& & + \ \Pr[(\bx_H, \by_H) \text{ is `good'}] \cdot \Ex_{\bx_H, \by_H} \insquare{\Ex_{r_A, r_B} \ f_2(\bx_H, r_A) \cdot g_2(\by_H, r_B) + \eps \bigg | (\bx_H, \by_H) \text{ is `good'}}\\
&=& \Pr[(\bx_H, \by_H) \text{ is not `good'}] \cdot \inparen{1 - \Ex_{\bx_H, \by_H} \insquare{\Ex_{r_A, r_B} \ f_2(\bx_H, r_A) \cdot g_2(\by_H, r_B) + \eps \bigg | (\bx_H, \by_H) \text{ is not `good'}}}\\
& & + \ \Ex_{\bx_H, \by_H} \insquare{\Ex_{r_A, r_B} \ f_2(\bx_H, r_A) \cdot g_2(\by_H, r_B) + \eps}\\
&\le& \Ex_{\bx_H, \by_H} \insquare{\Ex_{r_A, r_B} \ f_2(\bx_H, r_A) \cdot g_2(\by_H, r_B)} + 2\tau \cdot (2-\eps) + \eps\\
&\le& \Ex_{\bx_H, \by_H} \insquare{\Ex_{r_A, r_B} \ f_2(\bx_H, r_A) \cdot g_2(\by_H, r_B)} + 2 \eps
\end{eqnarray*}

\noindent Step 3 above is due to the definition of $f_2$ and $g_2$ and Theorem~\ref{th:inv_princ}. The last step follows because $\tau \ll \eps$, and so we can upper bound $2\tau \cdot (2-\eps) \le \eps$.\\

\noindent Thus, finally we choose $\eps = \gamma/2$ for Theorem~\ref{th:inv_princ}, and we get $\tau = \tau(\gamma)$ accordingly, thereby getting the final requirement of Lemma~\ref{lem:app_corr_bds}, that is,

$$\expect_{\substack{(\bx_H,\by_H)\sim \mu^{\otimes h} \\ (r_A, r_B) \sim \calG(\rho)}} f_2(\bx_H, r_A) \cdot g_2(\by_H, r_B) \quad \ge \quad \expect_{(\bx,\by) \sim \mu^{\otimes n}} f_1(\bx) \cdot g_1(\by) - \gamma$$
\end{proofof}
\fi
\section{Simulating Correlated Gaussians}
\noapx{\label{sec:gaussians}}{\label{apx:gaussians}}

\ifnum\apx=1
This section contains the full proof of Lemma~\ref{lem:app_witsenhausen_rounding}.
\else
In this section, we use the technique due to Witsenhausen \cite{witsenhausen1975sequences} which shows that for any joint probability space $(\calA \times \calB, \mu)$ with maximal correlation $\rho$, Alice and Bob can non-interactively simulate $\rho$-correlated gaussians upto arbitrarily small {\em 2-dimensional Kolmogorov distance}. We obtain the following lemma\fullver{ (proof in Appendix~\ref{apx:gaussians})}{}.

\begin{lem}[Witsenhausen's rounding] \label{lem:app_witsenhausen_rounding}
Let $(\calA \times \calB, \mu)$ be a joint probability space, and let $\rho = \rho(\calA, \calB; \mu)$ be its maximal correlation. Let $\zeta > 0$ be any given parameter. Then, there exists $w \defeq w((\calA \times \calB, \mu), \zeta) \in \bbN$, such that the following holds:

For all functions $f_2 : \calA^h \times \bbR \to [-1,1]$ and $g_2 : \calB^h \times \bbR \to [-1,1]$ having the following special form: there exist functions $f_2' : \calA^h \to \bbR$ and $g_2' : \calB^h \to \bbR$ such that,

$$f_2(\bx, r) = \infork{1 & r \ge f_2'(\bx) \\ -1 & r < f_2'(\bx)} \quad \sAND \quad g_2(\by, r) = \infork{1 & r \ge g_2'(\by) \\ -1 & r < g_2'(\by)}$$

\noindent there exist functions $f_3 : \calA^{h+w} \to [-1,1]$ and $g_3 : \calB^{h+w} \to [-1,1]$, such that,

$$\inabs{\Ex\limits_{\bx \sim \mu_A^{\otimes (h+w)}} f_3(\bx) - \Ex\limits_{\substack{\bx \sim \mu_A^{\otimes h} \\ r_A \sim \calN(0,1)}} [f_2(\bx, r_A)]} \le \zeta \quad \sAND \quad \inabs{\Ex\limits_{\by \sim \mu_B^{\otimes (h+w)}} g_3(\by) - \Ex\limits_{\substack{\bx \sim \mu_B^{\otimes h} \\ r_B \sim \calN(0,1)}} [g_2(\by, r_B)]} \le \zeta$$

\noindent and,

$$\inabs{\Ex\limits_{(\bx, \by) \sim \mu^{\otimes (h+w)}} [f_3(\bx) \cdot g_3(\by)] - \Ex\limits_{\substack{(\bx,\by) \sim \mu^{\otimes h} \\ (r_A, r_B) \sim \calG(\rho)}} [f_2(\bx, r_A) \cdot g_2(\by, r_B)]} \quad \le \quad \zeta$$

\noindent In particular, one may take $w = O\inparen{\frac{1+\rho}{\alpha \cdot (1-\rho)^{3} \cdot \zeta^2}}$, where $\alpha \defeq \alpha(\mu)$ is the minimum non-zero probability in $\mu$.
\end{lem}
\fi
\ifnum\stocapx=2 \else
\noindent The main idea in obtaining the functions $f_3$ and $g_3$ is the technique of Witsenhausen \cite{witsenhausen1975sequences}, of simulating $\rho$-correlated gaussians from many copies of $(\calA \times \calB, \mu)$.

\begin{lem}[Simulating gaussians \cite{witsenhausen1975sequences}] \label{lem:witsenhausen}
Let $(\calA \times \calB, \mu)$ be a joint probability space, and let $\rho = \rho(\calA, \calB; \mu)$ be its maximal correlation. Let $\zeta > 0$ be any given parameter. Then, there exists $w \defeq w((\calA \times \calB, \mu), \zeta) \in \bbN$, such that the following holds,

For all $\nu_1, \nu_2 \in [-1,+1]$, there exist functions $P_{\nu_1} : \calA^w \to [-1,1]$ and $Q_{\nu_2} : \calB^w \to [-1,1]$ such that $|\Ex[P_{\nu_1}(\bx)] - \nu_1| \le \zeta/2$, $|\Ex[Q_{\nu_2}(\by)] - \nu_2| \le \zeta/2$ and
$$ \inabs{\Ex_{(\bx,\by) \sim \mu^{\otimes w}} \ [P_{\nu_1}(\bx) Q_{\nu_2}(\by)] - \overline{\Gamma}_{\rho}(\nu_1,\nu_2)} \le \zeta$$

\noindent In particular, one may take $w = O\inparen{\frac{1+\rho}{\alpha \cdot (1-\rho)^{3} \cdot \zeta^2}}$, where $\alpha \defeq \alpha(\mu)$.
\end{lem}

\begin{proof}
Since $\rho = \rho(\calA, \calB; \mu)$, we have that there exist functions $f : \calA \to \bbR$ and $g : \calB \to \bbR$ such that $\Ex_{x \sim \mu_A} f(x) = \Ex_{y \sim \mu_B} g(y) = 0$, $\Var(f) = \Var(g) = 1$ and $\Ex_{(x,y) \sim \mu} [f(x) \cdot g(y)] = \rho$.

\noindent We define $F(\bx) = \frac{\sum_{i=1}^w f(x_i)}{\sqrt{w}}$ and $G(\by) = \frac{\sum_{i=1}^w g(y_i)}{\sqrt{w}}$. And define $P_{\nu_1}$ and $Q_{\nu_2}$ as follows,
$$P_{\nu_1}(\bx) = \infork{1 & F(\bx) \le \Phi^{-1}(\frac{1+\nu_1}{2}) \\ -1 & \text{otherwise}} \quad \sAND \quad Q_{\nu_2}(\by) = \infork{1 & G(\by) \le \Phi^{-1}(\frac{1+\nu_2}{2}) \\ -1 & \text{otherwise}}$$

\noindent We apply Lemma~\ref{lem:2d_CLT} for the pair of random variables $(f(x), g(y))$ with parameter $\zeta$ being $\zeta/4$, to obtain the appropriate $w$. It easily follows that, $|\Ex [P_{\nu_1}(\bx)] - \nu_1| \le \zeta/2$ and $|\Ex [Q_{\nu_2}(\by)] - \nu_2| \le \zeta/2$ and

$$\inabs{\Ex_{(\bx,\by) \sim \mu^{\otimes w}} \ [P_{\nu_1}(\bx) Q_{\nu_2}(\by)] - \overline{\Gamma}_{\rho}(\nu_1,\nu_2)} \le \zeta $$
\end{proof}

\noindent We are now ready to prove Lemma~\ref{lem:app_witsenhausen_rounding}.\\

\begin{proofof}{Lemma~\ref{lem:app_witsenhausen_rounding}}
Given $(\calA \times \calB, \mu)$ and $\zeta$, we obtain $w$ as in Lemma~\ref{lem:witsenhausen}. Given functions $f_2$ and $g_2$, of the said form, we construct functions $f_3 : \calA^{h+w} \to [-1,1]$ and $g_3 : \calB^{h+w} \to [-1,1]$ by invoking Lemma~\ref{lem:witsenhausen} for every assignment to the first $h$ variables with parameter $\zeta$. In particular for every $\bx_1 \in \calA^h, \bx_2 \in \calA^w$, we define $f_3(\bx_1, \bx_2) = P_{f_2'(\bx_1)}(\bx_2)$. Similarly, for $\by_1 \in \calB^h, \by_2 \in \calA^w$, we define $g_3(\by_1, \by_2) = Q_{g_2'(\by_1)}(\by_2)$.

\noindent This gives us that $|\Ex [f_3(\bx)] - \Ex[f_2(\bx, r_A)]| \le \zeta/2$ and $|\Ex [g_3(\by)] - \Ex[g_2(\by, r_B)]| \le \zeta/2$ and,

$$\inabs{\Ex\limits_{(\bx, \by) \sim \mu^{\otimes (h+w)}} [f_3(\bx) \cdot g_3(\by)] - \Ex\limits_{\substack{(\bx,\by) \sim \mu^{\otimes h} \\ (r_A, r_B) \sim \calG(\rho)}} [f_2(\bx, r_A) \cdot g_2(\by, r_B)]} \le \zeta$$

\noindent Thus, we have $f_3$ and $g_3$ as desired.
\end{proofof}
\fi
\ifnum\apx=0
\section{Putting it all together!} \label{sec:together}

In this section we finally use all the lemmas we have developed to prove Theorem~\ref{thm:main-lemma}.\\

\begin{proofof}{Theorem~\ref{thm:main-lemma}}
Given $(\calA \times \calB, \mu)$ and $\delta > 0$ and functions $f : \calA^n \to [-1,1]$ and $g : \calB^n \to [-1,1]$, we wish to apply Lemma~\ref{lem:app_corr_bds} with parameter $\gamma = \delta/3$ followed by Lemma~\ref{lem:app_witsenhausen_rounding} with parameter $\zeta = \delta/3$. Lemma~\ref{lem:app_corr_bds} will dictate a value $\tau = \tau((\calA \times \calB, \mu), \gamma)$. We wish to apply the Joint regularity lemma (Lemma~\ref{lem:joint_reg}), with this parameter $\tau$, which will dictate a value of $\eta = \eta(\tau)$. Using this value of $\eta$, and $\lambda = \delta/3$, we apply the Smoothing lemma (Lemma~\ref{lem:smoothing}), which will dictate a value of $d = d((\calA \times \calB, \mu), \lambda, \eta)$. We use this $d$ to feed into the joint regularity lemma (Lemma~\ref{lem:joint_reg}), to obtain a value of $h$. The final value of $n_0$ is the sum of $h((\calA \times \calB, \mu), d, \tau)$ given by the joint regularity lemma (Lemma~\ref{lem:joint_reg}) and $w((\calA \times \calB, \mu), \zeta)$ given by Witsenhausen's rounding procedure (Lemma~\ref{lem:app_witsenhausen_rounding}). This dependency of parameters is pictorially described in Figure~\ref{fig:dependency} (the dependencies on $(\calA \times \calB, \mu)$ are suppressed, for sake of clarity). It can be shown by putting everything together that $n_0 = \exp\inparen{\poly\inparen{\frac{1}{\delta}, \ \frac{1}{1-\rho}, \ \log\inparen{\frac{1}{\alpha}}}}$.

\fullver{}{\begin{figure}
\begin{center}
\begin{tikzpicture}[scale=0.8, transform shape]

\node[box, fill=blue!20] (delta) at (3.5,3) {\large $(\calA \times \calB, \mu), \delta$};

\node[box] (corr) at (8, 6) {\shortstack{Correlation Bounds \\ (Lemma~\ref{lem:app_corr_bds})}}
edge[in=50, out=-130, latex-]  node[above, midway] {$\gamma = \frac{\delta}{3} \ \ $} (delta);

\node[box] (wits) at (2, 6) {\shortstack{Witsenhausen Rounding \\ (Lemma~\ref{lem:app_witsenhausen_rounding})}}
edge[in=110, out=-90, latex-]  node[left, midway] {$\zeta = \frac{\delta}{3} \ \ $} (delta);

\node[box] (reg) at (6, 0) {\shortstack{Joint Regularity Lemma\\ (Lemma~\ref{lem:joint_reg})}}
edge[in=-50, out=110, latex-] (delta)
edge[in=-90, out=80, latex-] node[right, midway] {$\tau = \tau(\gamma)$} (corr);

\node[box] (smoot) at (0, 1.5) {\shortstack{Smoothing \\ (Lemma~\ref{lem:smoothing})}}
edge[in=180, out=90, latex-]  node[above, midway] {$\lambda = \frac{\delta}{3}\quad \quad $} (delta)
edge[in=150, out=0, latex-]  node[above, midway] {$\eta = \eta(\tau)$} (reg)
edge[in=180, out=-90, -latex]  node[above, midway] {$\ \ \ d = d(\lambda, \eta)$} (reg);

\node[box, fill=red!20] (ans) at (-4,0) {\large $n_0 = h+w$} 
edge[in=210, out=0, latex-]  node[below, midway] {$h = h(d, \tau)\quad \quad$} (reg)
edge[in=180, out=90, latex-]  node[right, midway] {$w = w(\zeta)$} (wits);

\end{tikzpicture}
\caption{Dependency of parameters in the proof of Theorem~\ref{thm:main-lemma}}
\label{fig:dependency}
\end{center}
\end{figure}
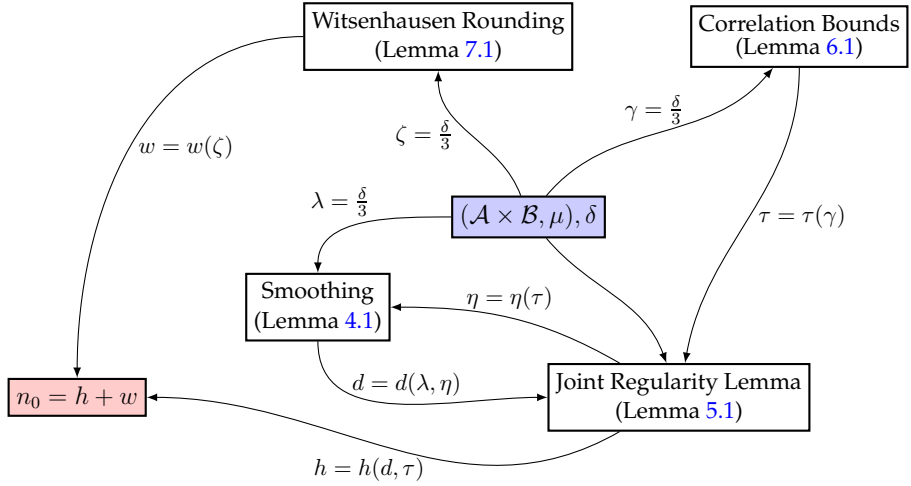}

\noindent Once we have all the parameters set, we are now able to apply them to any pair of functions $f : \calA^n \to [-1,1]$ and $g : \calB^n \to [-1,1]$. In particular, we proceed as described in the overview (Section~\ref{sec:overview}).

\begin{itemize}
\item [{\bf Step 1:}] We apply Lemma~\ref{lem:smoothing} to functions $f$ and $g$ with parameters $\lambda$ and $\eta$ as obtained above. This gives us a degree $d$ and functions $f_1$ and $g_1$, such that, $\sum_{|\bsigma| > d} \what{f}(\bsigma)^2 < \eta$ and $\sum_{|\bsigma| > d} \what{g}(\bsigma)^2 < \eta$.

\item [{\bf Step 2:}] We apply the joint regularity lemma (Lemma~\ref{lem:joint_reg}) on functions $f_1$ and $g_1$, with parameters $d$ and $\tau$ as obtained above (note that, the conditions involving $\eta$ are satisfied, because we chose precisely this $\eta$ to be given to the Smoothing lemma). This gives us a subset $H \subseteq [n]$ such that $|H| \le h$ and with high probability over restrictions to this subset $H$, the restricted versions of both $f_1$ and $g_1$ have all individual influences to be at most $\tau$.

\item [{\bf Step 3:}] We apply the correlation bounds result (Lemma~\ref{lem:app_corr_bds}) to functions $f_1$ and $g_1$ (note that all the conditions involving $\tau$ are satisfied already because we chose precisely this $\tau$ to be given to the joint regularity lemma).

This gives us functions $f_2 : \calA^h \times \bbR \to [-1,1]$ and $g_2 : \calB^h \times \bbR \to [-1,1]$ of the form: there exist functions $f_2' : \calA^h \to \bbR$ and $g_2' : \calB^h \to \bbR$ such that,

$$f_2(\bx, r) = \infork{1 & r \ge f_2'(\bx) \\ -1 & r < f_2'(\bx)} \quad \sAND \quad g_2(\by, r) = \infork{1 & r \ge g_2'(\by) \\ -1 & r < g_2'(\by)}$$

\item [{\bf Step 4:}] Functions $f_2$ and $g_2$ are exactly in the form for which Lemma~\ref{lem:app_witsenhausen_rounding} is applicable, which we use with parameters $\zeta$ as obtained above. This gives us functions $f_3 : \calA^{h+w} \to [-1,1]$ and $g_3 : \calB^{h+w} \to [-1,1]$.
\end{itemize}

\noindent Note that, $\Ex f = \Ex f_1 = \Ex f_2$ and $\inabs{\Ex f_3 - \Ex f_2} \le \zeta = \delta/3$ and similarly $\Ex g = \Ex g_1 = \Ex g_2$ and $\inabs{\Ex g_3 - \Ex g_2} \le \zeta = \delta/3$. Moreover, we have from Lemmas~\ref{lem:app_witsenhausen_rounding}, \ref{lem:app_corr_bds} and \ref{lem:smoothing} that,
\begin{eqnarray*}
\Ex_{(\bx, \by) \sim \mu^{\otimes (h+w)}} \insquare{f_3(\bx) \cdot g_3(\by)}
&\ge& \Ex_{\substack{(\bx, \by) \sim \mu^{\otimes h}\\ (r_A, r_B) \sim \calG(\rho)}} \insquare{f_2(\bx) \cdot g_2(\by)} - \zeta\\
&\ge& \Ex_{\substack{(\bx, \by) \sim \mu^{\otimes n}}} \insquare{f_1(\bx) \cdot g_1(\by)} - \gamma - \zeta\\
&\ge& \Ex_{\substack{(\bx, \by) \sim \mu^{\otimes n}}} \insquare{f(\bx) \cdot g(\by)} - \lambda - \gamma - \zeta\\
&=& \Ex_{\substack{(\bx, \by) \sim \mu^{\otimes n}}} \insquare{f(\bx) \cdot g(\by)} - \delta
\end{eqnarray*}

\noindent Hence, taking $\wtilde{f} = f_3$ and $\wtilde{g} = g_3$, proves Theorem~\ref{thm:main-lemma}.

\end{proofof}
\fi

\ifnum\stocapx=2 \else
\noapx{
\subsection{Generalizing to arbitrary binary targets}\label{sec:proof_main_full}
}{
\section{Generalizing to arbitrary binary targets}\label{apx:proof_main_full}
}

We now give a proof sketch of Theorem~\ref{thm:main_full}. Even though this is not a black-box application of  Theorem~\ref{thm:main}, it follows the same proof steps. We highlight the main differences in this section.\\

\noindent We consider two cases, (I) $\Ex[UV] \ge \Ex[U] \cdot \Ex[V]$ and (II) $\Ex[UV] \le \Ex[U] \cdot \Ex[V]$.

\subsubsection*{Case (I) : $\quad \Ex[UV] \ge \Ex[U] \cdot \Ex[V]$}

We need to modify the $\ABMIP$ problem \ref{prob:ABMIP}, by replacing the conditions on $\inabs{\Ex[f(\bx)]}$ by $\inabs{\Ex[f(\bx)] - \Ex[U]}$, and similarly replacing the conditions on $\inabs{\Ex[g(\by)]}$ by $\inabs{\Ex[g(\by)] - \Ex[V]}$ and replacing $\rho$ by $\Ex[UV]$. The reduction between $\ANIS$ and $\ABMIP$ works in almost exactly the same way.

It is easy to see that using the main technical theorem~\ref{thm:main-lemma} and following the same proof as of Theorem~\ref{thm:main}, we also get decidability for $\ANIS((\calA \times \calB, \mu), (\calU \times \calV, \nu), \delta)$.

\subsubsection*{Case (II) : $\quad \Ex[UV] \le \Ex[U] \cdot \Ex[V]$}

As in the previous case, we need to modify the $\ABMIP$ problem \ref{prob:ABMIP}, by replacing the conditions on $\inabs{\Ex[f(\bx)]}$ by $\inabs{\Ex[f(\bx)] - \Ex[U]}$, and similarly replacing the conditions on $\inabs{\Ex[g(\by)]}$ by $\inabs{\Ex[g(\by)] - \Ex[V]}$. The condition on $\Ex[f(\bx) g(\by)]$ will however change as, $\Ex[f(\bx) g(\by)] \le \Ex[UV] + \delta$ in case (i) vs. $\Ex[f(\bx) g(\by)] \ge \Ex[UV] + 4\delta$ in case (ii). The reduction between $\ANIS$ and $\ABMIP$ works in almost exactly the same way.

The main difference in this case however is that, we want each of the steps to `increase' correlation by a small amount as opposed to `decrease' the correlation. In particular, the main condition in Theorem~\ref{thm:main-lemma} will change as follows,

$$\Ex_{(\bx, \by) \sim \mu^{\otimes n_0}} \insquare{\wtilde{f}(\bx) \cdot \wtilde{g}(\by)} \quad \le \quad \Ex_{(\bx, \by) \sim \mu^{\otimes n}} \insquare{f(\bx) \cdot g(\by)} + \delta$$

The steps of Smoothing (Lemma~\ref{lem:smoothing}) and Joint Regularity (Lemma~\ref{lem:joint_reg}) and Witsenhausen rounding (Lemma~\ref{lem:app_witsenhausen_rounding}) don't need any modification as they approximately preserve the correlation in both directions. However, in the step of applying Correlation Bounds (Lemma~\ref{lem:app_corr_bds}), we need to use the lower bound of $\underline{\Gamma}_{\rho}(\cdot, \cdot)$ instead of the upper bound of $\overline{\Gamma}_{\rho}(\cdot, \cdot)$. In particular, the lemma will change slightly resulting in functions such that,
$$\expect_{\substack{(\bx,\by)\sim \mu^{\otimes h} \\ (r_A, r_B) \sim \calG(\rho)}} \insquare{f_2(\bx, r_A) \cdot g_2(\bx, r_B)} \quad \le \quad \expect_{(\bx,\by) \sim \mu^{\otimes n}} \insquare{f_1(\bx) \cdot g_1(\by)} + \gamma$$
Additionally, $f_2$ and $g_2$ will have the following special form: there exist functions $f_2' : \calA^h \to \bbR$ and $g_2' : \calB^h \to \bbR$ such that,

$$f_2(\bx, r) = \infork{1 & r \ge f_2'(\bx) \\ -1 & r < f_2'(\bx)} \quad \sAND \quad g_2(\by, r) = \infork{-1 & r \ge g_2'(\by) \\ 1 & r < g_2'(\by)}$$

This structural difference in $f_2$ and $g_2$ affects the Witsenhausen Rounding step (Lemma~\ref{lem:app_witsenhausen_rounding}) slightly, but it is easy to see that the same proof strategy works.

It is also easy to see that using this modified main theorem (analog of Theorem~\ref{thm:main-lemma}) and following the same proof steps as of Theorem~\ref{thm:main}, we also get decidability for $\ANIS((\calA \times \calB, \mu), (\calU \times \calV, \nu), \delta)$ in this case.
\fi
\section{Open Questions} \label{sec:discussion}

In this work, we proved computable bounds on the non-interactive simulation of any $2 \times 2$ distribution. We now conclude with some interesting open questions.

The running time of our algorithm is at least doubly-exponential in the input size\footnote{For constant values of $\delta$ and $\rho$, the running time is doubly-exponential in $2^{\poly(\log{m})}$. Here we think of the input as a bipartite graph with $m$ edges. This follows because $\alpha \sim 1/m$.}. It would be very interesting to understand the computational complexity of the non-interactive simulation problem. We point out that the question of generating the best $\DSBS$ can be thought of as a tensored version of the following ``{\sc Min-Bipartite-Bisection}'' problem: We are given a weighted bipartite graph $G = (L \cup R,E)$, and we wish to find a subset $S$ of $L \cup R$ such that $S \cap L$ roughly contains half the vertices of $L$, and $S \cap R$ roughly contains half the vertices of $R$, while minimizing the total weight of edges crossing the cut $(S,\overline{S})$. While it follows from \cite{raghavendra2012reductions} that {\sc Min-Bipartite-Bisection} is hard to approximate, the same is not necessarily true about its tensored version.

Another interesting open question is to generalize our decidability results to larger alphabets, which seems to require new technical ideas. Indeed, our proof of Theorems~\ref{thm:main_thm} and \ref{thm:main_thm_full} relied on the fact that for $(X,Y)$ being correlated random Gaussians, the maximum possible agreement of any pair of $\pm 1$-valued functions $f(X)$ and $g(Y)$ is at most that of two appropriate dictator threshold functions $F(X_1)$ and $G(Y_1)$ where $F$ only depends on the marginals of $f$ (i.e., the probability that $f$ takes the values $-1$ and $+1$), and similarly $G$ only depends on the marginals of $g$. The analogous statement for the ternary case is not true. Namely, let $f(X), g(Y) \in \{0,1,2\}$, and assume that the marginals of $f$ are $(1/3,1/3,1/3)$. Then, depending on whether the marginals of $g$ are $(1/3,1/3,1/3)$ or $(1/2,1/2,0)$, the largest agreement of $(f,g)$ would be achieved by very different functions $f$, assuming the ``Standard Simplex Conjecture'' (see \cite{isaksson2012maximally} and Proposition $2.10$ of \cite{HeilmanMN15}). This example shows that in the ternary case Alice cannot replace $f$ by a function of a very small number of copies without taking the marginals of Bob's function $g$ into account, and this is a major obstacle in generalizing our approach for proving Theorems~\ref{thm:main_thm} and \ref{thm:main_thm_full} to larger alphabets.


Yet another interesting open question is to generalize our computability results to more than two players, which also seems to require new technical ideas.

Finally, it will be very interesting to see if these techniques could apply to other `tensored' problems. The most relevant problems seem to be (i) deciding a quantum version of our problem, namely that of local state transformation of quantum entanglement \cite{Beigi_QuantumMaximalCorrelation, DelgoshaBeigi_QuantumHypercontractivity} and (ii) approximately computing the entangled value of a 2-prover 1-round game (\cite{KKMTV_HardnessApprox_EntangledGames}; also see the open problem \cite{openQIwiki_all_bell_inequalities}).

\section{Acknowledgments}
We thank Sudeep Kamath for explaining to us the state-of-the-art results in the information theory community, with regards to the problem of non-interactive simulation. We also thank Boaz Barak, Mohammad Bavarian and Mohsen Ghaffari for helpful discussions. We thank Matthew Coudron and Robin Kothari for pointing us to the related problems in the quantum literature.

%
%
%

\newcommand{\etalchar}[1]{$^{#1}$}

\end{document}